\documentclass[aps,pra,twocolumn,floatfix,superscriptaddress,longbibliography,notitlepage]{revtex4-2}

\usepackage{changes}

\setdeletedmarkup{\color{gray}\sout{#1}}

\usepackage[makeroom]{cancel}

\usepackage{mathrsfs}
\usepackage[normalem]{ulem}
\usepackage{graphicx, color, graphpap}
\usepackage{enumitem}
\usepackage{amssymb}
\usepackage{amsthm}
\usepackage{multirow}
\usepackage[colorlinks=true,citecolor=blue,linkcolor=magenta]{hyperref}
\usepackage[T1]{fontenc}
\usepackage{verbatim}
\usepackage{mathtools}
\usepackage{titlesec}

\long\def\ca#1\cb{} 

\newcommand{\abs}[2][]{#1| #2 #1|}

\newcommand{\bramatket}[3]{\langle #1 \hspace{1pt} | #2 | \hspace{1pt} #3 \rangle}

\newcommand{\ket}[1]{|#1\rangle}               
\newcommand{\bra}[1]{\langle #1|}              
\newcommand{\dya}[1]{\ket{#1}\!\bra{#1}}

\newcommand{\rank}{\text{rank}}

\newcommand{\Tr}{{\rm Tr}}

\newcommand{\diag}{\rm diag}
\newcommand{\supp}{\text{supp}}

\newcommand{\ave}[1]{\langle #1\rangle}               
\renewcommand{\geq}{\geqslant}
\renewcommand{\leq}{\leqslant}

\newcommand{\ad}{^\dagger}

\newcommand*{\id}{\openone}

\newtheorem{theorem}{Theorem}

\newtheorem{corollary}{Corollary}

\newtheorem*{theorem*}{Theorem}
\newtheorem*{corollary*}{Corollary}

\begin{document}
\title{{Jarzynski-like Equality of Nonequilibrium Information Production Based on Quantum Cross Entropy}}

\author{Akira Sone}
\email{akira.sone@umb.edu}
\affiliation{Department of Physics, University of Massachusetts, Boston, MA 02125, USA}
\affiliation{Aliro Technologies, Inc. Boston, Massachusetts 02135, USA}

\author{Naoki Yamamoto}
\affiliation{Quantum Computing Center, Keio University, 3-14-1 Hiyoshi,
Kohoku-ku, Yokohama, Kanagawa 223-8522, Japan}
\affiliation{
Department of Applied Physics and Physico-Informatics,
Keio University, Hiyoshi 3-14-1, Kohoku-ku, Yokohama 223-8522, Japan}

\author{Tharon Holdsworth}
\affiliation{Department of Physics, University of Massachusetts, Boston, MA 02125, USA}
\affiliation{Department of Physics, University of Alabama at Birmingham, Birmingham, Alabama 35294, USA}

\author{Prineha Narang}
\affiliation{College of Letters and Science, UCLA, Los Angeles, CA 90095 USA}
\affiliation{School of Engineering and Applied Sciences, Harvard University, Cambridge, Massachusetts 02138, USA.}

\begin{abstract}
{The two-time measurement scheme is well studied in the context of quantum fluctuation theorem. However, it becomes infeasible when the random variable determined by a single measurement trajectory is associated with the von-Neumann entropy of the quantum states. We employ the one-time measurement scheme to derive a Jarzynski-like equality of nonequilibrium information production by proposing an information production distribution based on the quantum cross entropy. The derived equality further enables one to explore the roles of the quantum cross entropy in quantum communications, quantum machine learning and quantum thermodynamics.}
\end{abstract}

\maketitle

\section{Introduction}
\label{sec:Intro}

Quantum thermodynamics {explores} the laws of the thermodynamics in the nanoscale from the perspective of the quantum information science~\cite{DeffnerBook19,Binder19,Anders16,sagawa2012thermodynamics,Goold16,ng2018resource,Auffeves2022energy,deffner2021energetic,aifer2022quantum,buffoni2022third}. 
 {On such scales}, statistical fluctuations become more significant, {and have principally been accounted for by} {fluctuation theorem~\cite{evans1993probability,Jarzynski97,crooks1999entropy,hatano2001steady, de2022quantum}}. {The discovery of the fluctuation theorem is one of the most important accomplishments in the thermodynamics to date}~\cite{Ortiz2011}.  {The} fluctuation theorem can be regarded as a first principle in thermodynamics, from which many fundamental {principles of} thermodynamic phenomena can be derived, such as arrow of time~\cite{jarzynski2013equalities} and response theory~\cite{andrieux2008quantum,andrieux2009fluctuation}.

{More recently, fluctuation theorem have equally been used to characterize information processing tasks.} For example, Sagawa and Ueda~\cite{Sagawa10}, and Fujitani and Suzuki~\cite{fujitani2010jarzynski}, related the fluctuation theorem with an efficacy of the feedback control for the manipulation of the total {entropy production} via measurements. The relation between the fluctuation theorem and the adiabaticity of the process was revealed by considering the state distinguishability~{\cite{Deffner16,Sone21b}}.  {In the context of} quantum computing and communications, Gardas and Deffner~\cite{gardas2018quantum} demonstrated that the  fluctuation theorem can be used to determine the dynamics of the quantum systems and the susceptibility to the thermal noise. Also, Kafri and Deffner~\cite{Kafri12} related the fluctuation theorem and the Holevo information~\cite{holevo1973bounds,holevo1998capacity,WatrousBook18}, {which upper bounds the amount of classical information that can be transmitted through the quantum channel.}

A standard approach to the fluctuation theorem in the quantum regime is the two-time measurement (TTM) scheme~\cite{Tasaki00,Kurchan01,Smith2018,An15,Campisi11,aguilar2021two,hernandez2020experimental,hernandez2021experimental,albash2013fluctuation,rastegin2014jarzynski}, in which the distribution of the measurement outcomes is constructed by the projection measurements on the system before and after the process. {The first measurement corresponds to the state preparation of the input state, which is an ensemble of the eigenstates of the first measurement weighted by the probabilities of obtaining the corresponding outcomes, while the second measurement can be independent from the output state~\cite{rastegin2013non,Kafri12}.}

While this scheme {corresponds} to the classical approach in stochastic thermodynamics~\cite{jarzynski2015quantum}, in the quantum regime, it is considered to be inconsistent because it does not taken into account the quantum coherence~\cite{perarnau2017no} and the informational contribution of back-action of projection  measurements~\cite{Deffner16}. {Particularly, for the information production (namely the von-Neumann entropy gain), one needs to fully obtain the information of the output state because the second measurement is strictly dependent on the principal components of the output state, which requires the quantum state tomography. Therefore, from the practical and conceptual perspective, the TTM scheme is infeasible when we want to deal with information production in the context of fluctuation theorem. Also, while there are other approaches beyond the TTM scheme, such as the Bayesian method~\cite{buscemi2021fluctuation,micadei2020quantum} and quasiprobabilty~\cite{levy2020quasiprobability,lostaglio2018quantum}, in order to deal with information production, they all strictly require the quantum state tomography; therefore, we need to find an alternative approach to deal with the information production.}

{To solve this problem, we employ so-called one-time measurement (OTM) scheme, which was proposed by Deffner, Paz and Zurek in Ref.~\cite{Deffner16}. In this scheme, similar to the TTM scheme, we perform a projection measurement initially, which corresponds to the state preparation of the input state. However, what differs from the TTM scheme is that the second projection measurement is avoided, so that the corresponding distribution of the measurement outcomes is determined by the conditional expectation of the observable of interest given the initial measurement outcome. This quantity can be estimated if the post-measurement state of the initial measurement and the dynamics are known. Particularly, for the information production, we do not have to diagonalize the output state in the OTM scheme, so that the OTM scheme is the \textit{only} option.}

{This paper is organized as the following. In Sec.~\ref{sec:MainResults}, we first propose an information production distribution for an input state of rank $r$ and a quantum channel. Then, we derive the Jarzynski-like equality and the lower bound on the total information production, which particularly becomes significant when we nee to consider the information flow of the system in the quantum processes~\cite{buscemi2016approximate,das2018fundamental}. We demonstrate that the lower bound is characterized by the quantum cross entropy. While there were less attentions on the quantum cross entropy, recently, the relations of the quantum cross entropy with the maximum likelihood principle in the machine learning~\cite{zhou_and_wang_2021quantum} and the quantum source coding~\cite{zhou2021quantum} have been explored. In our paper, we further explore the roles of quantum cross entropy in various protocols. In Sec.~\ref{sec:Example}, we discuss the applications of our result to quantum communications, quantum machine learning and quantum thermodynamics by focusing on the quantum autoencoder (QAE) protocol~\cite{Romero17,wan2017quantum}, which is a quantum data compression protocol assisted by the variational quantum algorithms (VQAs)~\cite{Preskill18, mcclean2016theory,jones2019variational,nakanishi2020sequential,cerezo2021variational,bharti2022noisy}, and the maximum available work theorem~\cite{deffner2017kibble} in the quantum thermodynamic systems, followed by the conclusion in Sec.~\ref{sec:conc}.}

\section{Main Results}
\label{sec:MainResults}

{Let us consider a Hilbert space $\mathcal{H}$ of dimension $d\equiv\dim(\mathcal{H})$. Let $\mathcal{B}(\mathcal{H})$ denote the set of the density matrices acting on $\mathcal{H}$. We initially prepare a quantum state $\rho_0\in\mathcal{B}(\mathcal{H})$, and perform a measurement with an observable $P\equiv \sum_{i=1}^{d}a(p_i)\Pi_{i}$, where $\Pi_i\equiv\dya{p_i}$ are the projectors on the eigenbases of $P$. Suppose that the outcome is $a(p_i)$. Then, the post-measurement state is given by $\dya{p_i}=\Pi_i\rho_0 \Pi_i/p_i$ with $p_i\equiv \Tr\left[\rho_0\Pi_i\right]$. Then, in general, the input state is given by~\cite{Kafri12,rastegin2013non}}
\begin{align}
    {\rho_{\text{in}}=\sum_{i=1}^{r}p_i\dya{p_i}\,,}
\label{eq:input}
\end{align}
{where $r\equiv\rank(\rho_{\text{in}})$ denotes the rank of the input state. The state $\rho_{\text{in}}$ is an ensemble of the eigenbases of the initial measurement $P$ weighted by the probabilities of obtaining the outcomes $a(p_i)$; therefore, we can regard the initial measurement as a protocol of the state preparation of $\rho_{\text{in}}$. In this case, $p_i$ satisfies the following conditions $0<p_i\leq 1~(1\leq i\leq r)$, $p_i=0~(r+1\leq i\leq d)$, and   $\sum_{i=1}^{r}p_i=1$.}

{Let $\Phi:\mathcal{B}(\mathcal{H})\to\mathcal{B}(\mathcal{H}')$ be a quantum channel, which is a completely positive and trace-preserving (CPTP) map~\cite{wilde2013quantum}. Through this channel, the output state $\rho_{\text{out}}\in\mathcal{B}(\mathcal{H}')$ is given by
\begin{align}
\rho_{\text{out}}\equiv \Phi(\rho_{\text{in}})\,.
\label{eq:Output}
\end{align}
The total information production is defined as
\begin{align}
    \Delta S \equiv S(\rho_{\text{out}})-S(\rho_{\text{in}})\,,
\label{eq:InformationProduction}
\end{align}
where $S(\rho)\equiv-\Tr\left[\rho\ln\rho\right]$ denotes the von-Neumann entropy of the quantum state $\rho$.}

Here, we propose the following   {information production} distribution in the OTM scheme~\footnote{{The OTM scheme has been utilized to explore work and heat in the open quantum system~\cite{Sone20a}, its classical correspondence~\cite{Sone21b}, heat exchange~\cite{sone2022heat}, and work as an external observable~\cite{Beyer2020}. Particularly, a second-law-like inequality involving the guessed heat introduced in Ref.~\cite{Sone20a} can be derived by using $\widetilde{P}(\sigma)$ (See Appendix.~\ref{app:2ndLawGuessedHeat}).}}
\begin{align}
    \widetilde{P}(\sigma)\equiv\sum_{i=1}^{r}p_i\delta(\sigma-C(\Phi(\dya{p_i}),\rho_{\text{out}})-\ln p_i)\,,
\label{eq:DisOTM}
\end{align}
where $C(\rho_1,\rho_2)\equiv-\Tr\left[\rho_1\ln\rho_2\right]$ denotes the quantum cross entropy of $\rho_{1}$ with respect to $\rho_2$. {Let $\supp(\rho)$ denote the support of a quantum state $\rho$.}  {Then, note that}  $C(\rho_1,\rho_2)<\infty$ ($\supp(\rho_1)\subseteq\supp(\rho_2)$) and $C(\rho_1,\rho_2)=\infty$ (otherwise). Also, by definition, we have $C(\rho,\rho)=S(\rho)$. In Eq.~\eqref{eq:DisOTM}, {due to $\rho_{\text{out}}=\Phi(\rho_{\text{in}})=\sum_{i=1}^{r}p_i\Phi(\dya{p_i})$}, we have $\supp(\Phi(\dya{p_i}))\subseteq\supp(\rho_{\text{out}})$, so that $C(\Phi(\dya{p_i}),\rho_{\text{out}})<\infty$.  

{With this distribution, the average of $\sigma$ with respect to $\widetilde{P}(\sigma)$ becomes the exact information production}
\begin{align}
\ave{\sigma}_{\widetilde{P}} = S(\rho_{\text{out}})-S(\rho_{\text{in}})=\Delta S\,,
\label{eq:EntropyProdOTM}
\end{align}
{where we used the linearity on the first argument of the quantum cross entropy~\cite{zhou_and_wang_2021quantum} and Eq.~\eqref{eq:Output}. Let $r'\equiv\rank(\rho_{\text{out}})$ be the rank of the output state. Then, we can interpret the random variable $\sigma$ as follows. Let $\{q_j, \ket{q_j}\}_{j=1}^{r'}$ denote an eigensystem of $\rho_{\text{out}}$. Let us define the transition probability $P(j|i)\equiv\bramatket{q_j}{\Phi(\dya{p_i})}{q_j}$. Then, in the OTM scheme, $\sigma$ randomly takes $\sum_{j}(-\ln  q_j)P(j|i)+\ln p_i=\sum_{j}(-\ln q_j+\ln p_i)P(j|i)$, which is the conditional expectation of the information production given the initial measurement outcome. }

Here, $\sigma$ can be also identified to be a random variable as a source of the {information production} $\Delta S$.
Therefore, $\widetilde{P}(\sigma)$ is a \textit{good} definition. {Averaging the exponentiated information production with respect to the distribution in Eq.~\eqref{eq:DisOTM},} we can obtain our main result:
\begin{theorem}[Jarzynski-like equality of nonequilibrium information production]
\label{th:JarzynskiOTM}
The {Jarzynski-like equality of nonequilibrium information production} is 
\begin{align}
    \ave{e^{-\sigma}}_{\widetilde{P}}= \sum_{i=1}^{r}e^{-C(\Phi(\dya{p_i}),\rho_{\text{out}})}\,,
\label{eq:JarzynskiOTM}
\end{align}
which results in 
\begin{equation}
{\Delta S\geq L_{\text{otm}}\,,}
\label{eq:LowerBound}
\end{equation}
{where $L_{\text{otm}}$ is defined as }
\begin{align}
    {L_{\text{otm}}\equiv -\ln\left(\sum_{i=1}^{r}e^{-C(\Phi(\dya{p_i}),\rho_{\text{out}})}\right)\,. }
    \label{eq:LotmDef}
\end{align}
\end{theorem}
\begin{proof}
From Eqs.~\eqref{eq:input} and \eqref{eq:DisOTM}, we have
\begin{align}
\begin{split}
    \ave{e^{-\sigma}}_{\widetilde{P}} =& \int d\sigma \widetilde{P}(\sigma)e^{-\sigma}\\
    =&\sum_{i=1}^{r}p_ie^{-C(\Phi(\dya{p_i}),\rho_{\text{out}})} e^{-\ln p_i}\\
    =&\sum_{i=1}^{r}e^{-C(\Phi(\dya{p_i}),\rho_{\text{out}})}\,,
\end{split}
\end{align}
which proves Eq.~\eqref{eq:JarzynskiOTM}. From Jensen's inequality $\ave{e^{-\sigma}}_{\widetilde{P}}\geq e^{-\ave{\sigma}_{\widetilde{P}}}$ and Eq.~\eqref{eq:EntropyProdOTM}, we obtain {Eq.~\eqref{eq:LowerBound}} \footnote{{Note that our main claims are the derivation of the general integrated fluctuation theorems, which hold for any states and quantum channels, and its potential of characterizing the quantum protocol with the quantum cross entropy. We leave the comparison between the tight bound derived in Refs.~\cite{das2018fundamental,buscemi2016approximate} and our bound derived in the OTM scheme as an open problem.}}. 
\end{proof}

{When $\Phi:\mathcal{B}(\mathcal{H})\to\mathcal{B}(\mathcal{H})$ is particularly a unital map (i.e., $\Phi(\id)=\id$, where $\id$ denotes the identity matrix acting on $\mathcal{H}$), it well known that we have $\Delta S\geq 0$~\cite{Nielsen}. However, we can obtain a tighter bound as demonstrated in the following corollary}:
\begin{corollary}[Lower bound from OTM scheme under a unital map]
\label{cor:OTMUnitalMapBound}
When {$\Phi:\mathcal{B}(\mathcal{H})\to\mathcal{B}(\mathcal{H})$} is a unital map, $L_{\text{otm}}$ is  {a tighter bound on $\Delta S$ as}
\begin{align}
{\Delta S\geq L_{\text{otm}}\geq 0}\,.
\label{eq:LowerBoundUnitalMap}
\end{align}
\end{corollary}
\begin{proof}

{Given an input state $\rho_{\text{in}}=\sum_{i=1}^{r}p_i\dya{p_i}$ of rank $r$, let  $\Pi_{\text{in}}\equiv \sum_{i=1}^{r}\dya{p_i}~(\overline{\Pi}_{\text{in}}\equiv \id-\Pi_{\text{in}})$ be the projectors onto the support (null space) of $\rho_{\text{in}}$. Let $\Phi:\mathcal{B}(\mathcal{H})\to\mathcal{B}(\mathcal{H})$ be a unital map, i.e. $\Phi(\id)=\id$.} {Because the quantum cross entropy can be lower bounded by using the state overlap~\cite{zhou_and_wang_2021quantum}, we have}
\begin{align}
{C(\Phi(\dya{p_i}),\rho_{\text{out}})\geq -\ln\Tr\left[\Phi(\dya{p_i})\rho_{\text{out}}\right]\,.}   
\end{align}
{Then, due to the linearity of the CPTP map, we can obtain}
\begin{align}
    {\sum_{i=1}^{r}e^{-C(\Phi(\dya{p_i}),\rho_{\text{out}})}\leq \Tr\left[\Phi(\Pi_{\text{in}})\rho_{\text{out}}\right]\,.}
\label{eq:LotmUpperBound}
\end{align}
{Because $\Pi_{\text{in}}+\overline{\Pi}_{\text{in}}=\id$, from}
\begin{align}
    {\Phi(\id)=\Phi(\Pi_{\text{in}})+\Phi(\overline{\Pi}_{\text{in}})=\id\,,}
\label{eq:OverlapUpperbound}
\end{align}
{we can obtain}
\begin{align}
    {\Tr\left[\Phi(\Pi_{\text{in}})\rho_{\text{out}}\right] = 1-\Tr\left[\Phi(\overline{\Pi}_{\text{in}})\rho_{\text{out}}\right]\leq 1\,.}
\end{align}
{Therefore, from Eqs.~\eqref{eq:LotmDef}, \eqref{eq:LotmUpperBound} and \eqref{eq:OverlapUpperbound}, we obtain Eq.~\eqref{eq:LowerBoundUnitalMap}, which proves Corollary.~\ref{cor:OTMUnitalMapBound}.}
\end{proof}

\section{Examples}
\label{sec:Example}

In this section, we illustrate two applications of our {result}: quantum autoencoder and quantum thermodynamics.

\subsection{Quantum Autoencoder}
As our first example, we demonstrate the application of {our result in} the quantum autoencoder (QAE) proposed by Romero, Olson and Aspuru-Guzik in Ref.~\cite{Romero17}. The QAE is a quantum analogue of the (classical) variational autoencoder~\cite{Kingma_Book_Autoencoder}. In the QAE, the encoding and decoding operations are described by a parameterized quantum circuit. The original quantum data is compressed to the latent system by tracing over the other subsystem. Then, one prepares the fresh qubits, and decompressed the quantum data through the decoding operation acting on the fresh-qubit system and the latent system. The goal of the protocol is to recover the quantum data in the output, implying that a low-dimensional feature quantum state is well extracted through the encoding process; thus, we can use the resulting decoding process as a generative model to produce a quantum state outside the training quantum dataset by fluctuating the feature state. 
The cost function dependent on these tunable parameters, which measures the distance between the output and input state, is constructed by the quantum computer, and the set of the parameters is optimized through training the cost function with the classical computers. Recently, as a practical near-term quantum algorithm, the QAE has been widely explored both theoretically and experimentally~\cite{Bravo_Prieto_2021,Bondarenko2019,locher2022quantum,Cao_2021,steinbrecher2019quantum,du2021exploring,pepper2019experimental,mangini2022quantum,patel2022information,ngairangbam2022anomaly,huang2020realization,ma2020compression}.

{Let us describe the setup of the QAE below.} We consider a composite Hilbert space $\mathcal{H}=\mathcal{H}_A\otimes\mathcal{H}_B$, where $\mathcal{H}_A ~(\mathcal{H}_B)$ denotes the Hilbert space of the reduced quantum system $A~(B)$. For the followings, let us regard $\mathcal{H}_A$ as the latent Hilbert space, into which we compress our quantum data. Also, let us write $d_j$ as the dimension of the reduced Hilbert space $\mathcal{H}_j$, $d_j\equiv\dim(\mathcal{H}_j)~(j=A,B)$, so that the dimension of the total system is given by $d=d_Ad_B$. Following Ref.~\cite{Romero17}, we consider the following scenario (see Fig.~\ref{fig:QAE}).
\begin{figure}[htp!]
    \centering
    \includegraphics[width=1.0\columnwidth]{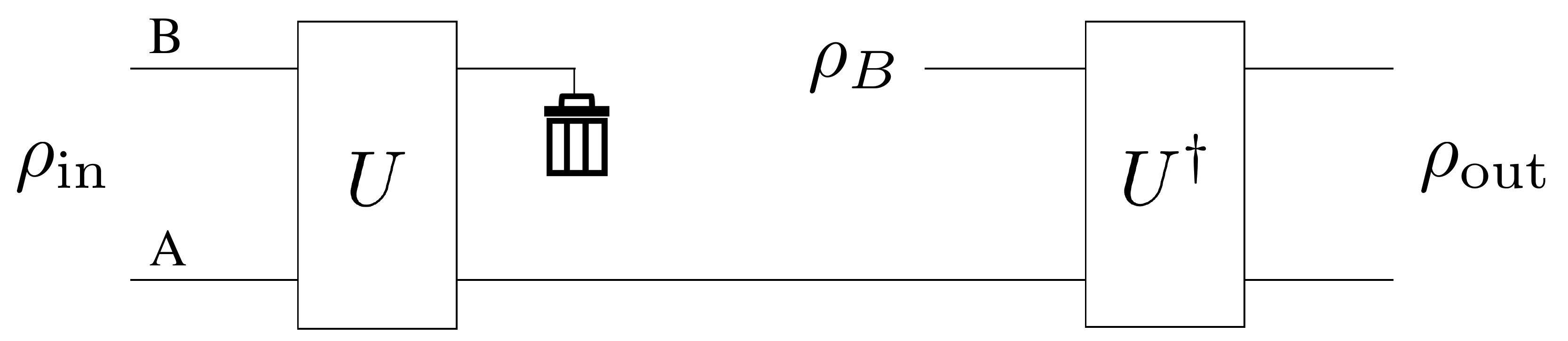}
    \caption{\textbf{Quantum Autoencoder:} We use $U$ to compress the input state $\rho_{\text{in}}$ into the reduced Hilbert space $\mathcal{H}_A$, and use the state $\rho_B$ of the fresh qubits in $\mathcal{H}_B$ to decompress the data by applying the unitary $U\ad$ to generate the output state $\rho_{\text{out}}$. }
    \label{fig:QAE}
\end{figure} 
In this setup, we apply a parameterized unitary $U$ to the input state $\rho_{\text{in}}$ and perform the partial trace over $\mathcal{H}_B$ to compress the quantum data into the latent Hilbert space $\mathcal{H}_A$. Then, we use the fresh qubits prepared in the state {$\rho_B\in\mathcal{B}(\mathcal{H}_B)$} to decompress the data by applying the unitary $U\ad$ to generate the output state $\rho_{\text{out}}$.  In this case, we have 
\begin{align}
\rho_{\text{out}}=\Phi(\rho_{\text{in}})\equiv U\ad\left(\Tr_B\left[U\rho_{\text{in}}U\ad\right]\otimes\rho_B\right)U\,.
\label{eq:QAEevolution}
\end{align}
{To discuss the Jarzynski-like equality, it is convenient to} define the compressed states
\begin{align}
    \rho_A&\equiv \Tr_B\left[U\rho_{\text{in}}U\ad\right]\label{eq:CompressedState1}\\
    \rho_A^{(i)}&\equiv \Tr_B\left[U\dya{p_i}U\ad\right]
    \label{eq:CompressedState2}\,.
\end{align}
Therefore, we can write $\rho_A = \sum_{i=1}^{d}p_i\rho_A^{(i)}$, so that we have $\supp(\rho_A^{(i)})\subseteq \supp(\rho_A)$.

{Given this setup above, we can relate $L_{\text{otm}}$ to the classical information transmission and the cost function in QAE, which demonstrates the roles of the quantum cross entropy in quantum communications and quantum machine learning in the framework of QAE protocol.}

{Let us first derive the expression of $L_{\text{otm}}$ in QAE. From Eq.~\eqref{eq:QAEevolution}, we have 
\begin{align}
C(\Phi(\dya{p_i}),\rho_{\text{out}})=S(\rho_B)+C(\rho_A^{(i)},\rho_A)\,.
\end{align}
Therefore, we can write 
\begin{align}
 \ave{e^{-\sigma}}_{\widetilde{P}}=e^{-S(\rho_B)}\sum_{i=1}^{r}e^{-C(\rho_A^{(i)},\rho_{A})}\,,   
\end{align}
so that $L_{\text{otm}}$ is given by 
\begin{align}
    L_{\text{otm}} = S(\rho_B)-\ln\left(\sum_{i=1}^{r}e^{-C(\rho_A^{(i)},\rho_{A})}\right)\,.
\label{eq:LotmQAE}
\end{align}
An important observation is that $\ave{e^{-\sigma}}_{\widetilde{P}}$ includes two terms which characterize the protocols of the QAE. One is the von-Neumann entropy $S(\rho_B)$, which is the informational contribution from the state preparation protocol in the fresh-qubit system $\mathcal{H}_B$. The other one is associated with the quantum cross entropy $C(\rho_A^{(i)},\rho_A)$ with respect to the latent Hilbert space $\mathcal{H}_A$. This quantity can be regarded as a term characterizing the compression protocol of the QAE. In the following, we explore the roles of the quantum cross entropy in the quantum communications and quantum machine learning from the relation of the lower bound $L_{\text{otm}}$ to the loss of Holevo information and the global cost function of the QAE.}

\subsubsection{{Relation to the loss of Holevo information in QAE}}

{Here, we explore the relation between $L_{\text{otm}}$ and the entropic disturbance. Entropic disturbance is the loss of Holevo information through a given quantum channel $\Phi$~\cite{buscemi2009towards,buscemi2016approximate}. Hence, it quantifies the loss of the maximum amount of classical information transmittable through the quantum channel. In Ref.~\cite{buscemi2016approximate}, a lower bound on $\Delta\chi$ was derived. Here, for a given quantum channel $\Phi$, we provide an \textit{upper} bound on the entropic disturbance by using $L_{\text{otm}}$ to provide a operational meaning to the quantum cross entropy in terms of classical information transmission.}

{The entropic disturbance is defined as follows. Let $\mathcal{L}\equiv\{p_i,\rho_i\}_{i=1}^{r}$ denote an ensemble of input state $\rho_{\text{in}}\equiv \sum_{i=1}^{r}p_i\rho_i$ and $\Phi(\mathcal{L})\equiv \{p_i,\Phi(\rho_i)\}_{i=1}^{r}$ denote the ensemble of output state $\rho_{\text{out}}\equiv\Phi(\rho_{\text{in}})=\sum_{i=1}^{r}p_i\Phi(\rho_i)$. Entropic disturbance is defined as $\Delta \chi \equiv \chi(\mathcal{L})-\chi(\Phi(\mathcal{L}))$, where $\chi(\mathcal{L})\equiv S(\rho_{\text{in}})-\sum_{i=1}^{r}p_i S(\rho_i)$ and $\chi(\Phi(\mathcal{L}))\equiv S(\rho_{\text{out}})-\sum_{i=1}^{r}p_iS(\Phi(\rho_i))$ are the Holevo information of $\rho_{\text{in}}$ and $\rho_{\text{out}}$, respectively. In our case, we have $\rho_i=\dya{p_i}$, so that $S(\rho_i)=S(\dya{p_i})=0$. Therefore, due to $\Delta S\geq L_{\text{otm}}$ and Eq.~\eqref{eq:LotmDef}, we can obtain
\begin{align}
\Delta \chi\leq \ln\left(\sum_{i=1}^{r}e^{-C(\Phi(\dya{p_i}),\rho_{\text{out}})}\right)+\sum_{i=1}^{r}p_i S(\Phi(\dya{p_i})),
\label{eq:EntropicDisturbance}
\end{align}
{which shows that the upper bound of the entropic disturbance can be characterized by the quantum cross entropy. }
\footnote{{When $\Phi$ is a unitary operation, we have $L_{\text{otm}}=0$ and $S(\Phi(\dya{p_i}))=0$. Because the entropic disturbance is invariant under the unitary operation $\Delta\chi=0$, we can say that the upper bound is tight for the unitary operation.}}.}

{Now, let us consider the case of QAE, in which $\Phi$ satisfies Eq.~\eqref{eq:QAEevolution}. Due to $S(\Phi(\dya{p_i}))=S(\rho_A^{(i)})+S(\rho_B)$, we have $\sum_{i=1}^{r}p_iS(\Phi(\dya{p_i}))=\sum_{i=1}^{r} p_iS(\rho_A^{(i)})+S(\rho_B)$, so that the upper bound on $\Delta \chi$ in QAE is given by
\begin{align}
    \Delta \chi \leq \sum_{i=1}^{r}p_i S(\rho_A^{(i)})+\ln\left(\sum_{i=1}^{r}e^{-C(\rho_A^{(i)},\rho_A)}\right)\,.
\end{align}
Therefore, the information and the quantum cross entropy of the compressed states contributes to setting an upper bound on entropic disturbance of in the QAE protocol. Also, note that for the QAE, $\Delta\chi$ can be explicitly written as
\begin{align}
\Delta\chi = S(\rho_{\text{in}})-S(\rho_A)-\sum_{i=1}^{r}p_iS(\rho_A^{(i)})\,,
\end{align}
which implies that the loss of the maximum amount of classical information in the QAE protocol is \textit{independent} of the choice of $\rho_B$ but strictly dependent on the compressed and input state. }

\subsubsection{{Relation to the global cost function of QAE}}
{The lower bound $L_{\text{otm}}$ can be also related to the performance of the QAE, which can be characterized by its global cost function. The cost function of the QAE is well-defined when the fresh-qubit state $\rho_B$ is a pure state 
\begin{align}
\rho_B=\dya{\psi}\,.
\label{eq:FreshQubitPure}
\end{align}
In this case, from Eq.~\eqref{eq:LotmQAE} and $S(\rho_B)=S(\dya{\psi})=0$. we have 
\begin{align}
    L_{\text{otm}}= -\ln\left(\sum_{i=1}^{r}e^{-C(\rho_A^{(i)},\rho_A)}\right)\,.
\label{eq:LotmQAEpure}
\end{align}
}

{Let $\eta_B$ denote the reduced state $\eta_B\equiv \Tr_A[U\rho_{\text{in}}U\ad]$. Then, from Refs.~\cite{Romero17,cerezo2021barren}, the global cost function $\mathscr{C}$ can be given by 
\begin{align}
    \mathscr{C}\equiv 1-\bramatket{\psi}{\eta_B}{\psi}\,,
    \label{eq:GlobalCostQAE}
\end{align}
which satisfies $0\leq \mathscr{C}\leq 1$. The ultimate goal of this protocol is to find a optimal unitary $U_{*}$ to realize $\rho_{\text{in}}=\rho_{\text{out}}$. In this optimal case, the global cost function is $\mathscr{C}=0$. Then, by using $d_B$ the dimension of $\mathcal{H}_B$ and the global cost function, we can obtain the following inequality (See Appendix.~\ref{app:LandGlobalCost} for the proof)
\begin{align}
    L_{\text{otm}}\leq\Delta S\leq 2\ln\left(\sqrt{1-\mathscr{C}}+\sqrt{(d_B-1)\mathscr{C}}\right)\,.
\label{eq:LandGlobalCost}
\end{align}
From Eq.~\eqref{eq:LotmQAEpure}, we can finally obtain
\begin{align}
    \sum_{i=1}^{r}e^{-C(\rho_A^{(i)},\rho_{A})}\geq\left(\frac{1}{\sqrt{1-\mathscr{C}}+\sqrt{(d_B-1)\mathscr{C}}}\right)^2\,,
\end{align}
which shows that the quantum cross entropy plays a role as an informational contribution of the compressed state to the performance of the QAE protocol.
Here, note that, due to $0\leq\mathscr{C}\leq 1$, we have $0\leq 2\ln\left(\sqrt{1-\mathscr{C}}+\sqrt{(d_B-1)\mathscr{C}}\right)\leq \ln(d_B)$, where $\mathscr{C}=0$ (i.e., $\eta_B=\dya{\psi}$) leads to the minimum, and $\mathscr{C}=1-1/d_B$ (i.e., $\eta_B=\id_B/d_B$) leads to the maximum. 
}

{We can also check the consistency of Eq.~\eqref{eq:LandGlobalCost} by considering the optimal case. The optimal unitary $U_{*}$ is a disentangling gate~\cite{ma2020compression}, so that $U_{*}$ satisfies 
\begin{align}
    \begin{split}
      U_{*}\rho_{\text{in}}U_{*}\ad&=\rho_A\otimes\dya{\psi}\\
        U_{*}\dya{p_i}U_{*}\ad &= \rho_{A}^{(i)}\otimes\dya{\psi}\,.
    \end{split}
\label{eq:OptUnitaryCondition}
\end{align}
Then, by using the unitary invariance of the quantum cross entropy, we obtain
\begin{align}
\begin{split}
C\left(\rho_A^{(i)},\rho_A\right)&=C\left(\rho_A^{(i)}\otimes\dya{\psi}, \rho_A\otimes\dya{\psi}\right)\\
&=C(\dya{p_i},\rho_{\text{in}})\\
&=-\ln p_i\,.
\end{split}
\end{align}
In this way, we have $L_{\text{otm}}=0$.
By definition, in the optimal case, we have $\mathscr{C}=0$; therefore, $\ln\left(\sqrt{1-\mathscr{C}}+\sqrt{(d_B-1)\mathscr{C}}\right)=0$. From $L_{\text{otm}}\leq \Delta S\leq 2\ln\left(\sqrt{1-\mathscr{C}}+\sqrt{(d_B-1)\mathscr{C}}\right)$, we get the expected result $\Delta S=0$ for the optimal unitary case. }

\subsection{{Maximum Available Work Theorem}}
\label{subsec:2ndLaw}
{For the second example, we explore the role of quantum cross entropy in work extraction from the quantum thermodynamic systems by from the relation between $L_{\text{otm}}$ and the maximum available work theorem~\cite{deffner2017kibble}}.

{In Ref.~\cite{deffner2017kibble}, a generic quantum thermodynamic system is regarded as a tripartite system composed of the system, work reservoir and heat bath. In this setup, the work $\ave{W}$, the internal energy change of the system $\Delta E_s$ and the internal energy change of the heat bath $\Delta E_b$ satisfy the first law of thermodynamics $\Delta E_s+\Delta E_b=\ave{W}$. Also, when $\Delta S$ and $\Delta S_b$ denote the von-Neumann entropy change of the system and heat bath, respectively, the second law of thermodynamics states $\Delta S+\Delta S_b\geq 0$. Because the heat reservoir is so large, which can be regarded as being always in equilibrium at inverse temperature $\beta$, we can write $\Delta S_b=\beta \Delta E_b$. Then, the maximum available work theorem states 
\begin{align}
\ave{W}\geq \Delta E_s-\beta^{-1}\Delta S\equiv\Delta \mathcal{E}\,,
\label{eq:MaximumAvailableWorkTheorem}
\end{align}
where $\mathcal{E}$ is called exergy or availability, which quantifies the maximally available work.}

{In this setup, from Theorem.~\ref{th:JarzynskiOTM}, we can obtain the upper bound on the exergy $\Delta \mathcal{E}$ as
\begin{align}
\Delta\mathcal{E}\leq \Delta E_s+\beta^{-1}\ln\left(\sum_{i=1}^{r}e^{-C(\Phi(\dya{p_i}),\rho_{\text{out}})}\right)\,.
\end{align}
This demonstrates the informational contribution of the quantum cross entropy in extracting maximally available work in the quantum thermodynamic systems. If there is no work reservoir, i.e., $\ave{W}=0$, the corresponding maximum available work theorem becomes $\Delta S\geq -\beta \Delta E_b$. However, when the system undergoes the energy-emitting process ($\Delta E_b\geq 0$) described by a unital evolution, from Corollary~\ref{cor:OTMUnitalMapBound}, we have a tighter bound as 
\begin{align}
\Delta S\geq L_{\text{otm}}\geq -\beta\Delta E_b\,.
\label{eq:TighterBoundEnergyExchange}
\end{align}
}

{A good example of the energy-emitting unital evolution is the spin-boson model~\cite{schlosshauer2007decoherence}.} Let us consider a system $\mathcal{H}_s$ initially prepared in  $\rho_{\text{in}}$ coupled to a heat bath $\mathcal{H}_b$, whose initial state is prepared in the Gibbs state
\begin{align}
\rho_b^{\text{eq}}\equiv\frac{e^{-\beta H_b}}{Z}\,,
\end{align}
where $Z\equiv\Tr\left[e^{-\beta H_b}\right]$ is the canonical partition function with inverse temperature $\beta$ and $H_b$ the time-independent bare Hamiltonian of the bath. Then, when $\Phi$ is a thermal operation~\cite{sagawa2012thermodynamics,ng2018resource,Goold16} from $t=0$ to $t=\tau$,
\begin{align}    \rho_{\text{out}}=\Phi(\rho_{\text{in}})=\Tr_b\left[U_{\tau}(\rho_{\text{in}}\otimes\rho_b^{\text{eq}})U_{\tau}\ad\right]\,.
\label{eq:SpinBosonThermalOp}
\end{align}

{In the quantum thermodynamic setup of the spin-boson model, $\mathcal{H}_s$ and $\mathcal{H}_b$ usually describes a two-level atomic system and the bosonic heat bath, respectively (See Fig.~\ref{fig:spin-boson}), and the atomic system in Eq.~\eqref{eq:SpinBosonThermalOp} undergoes the dephasing process, which is described by a unital map.}
\begin{figure}[htp!]
    \centering
    \includegraphics[width=.7\columnwidth]{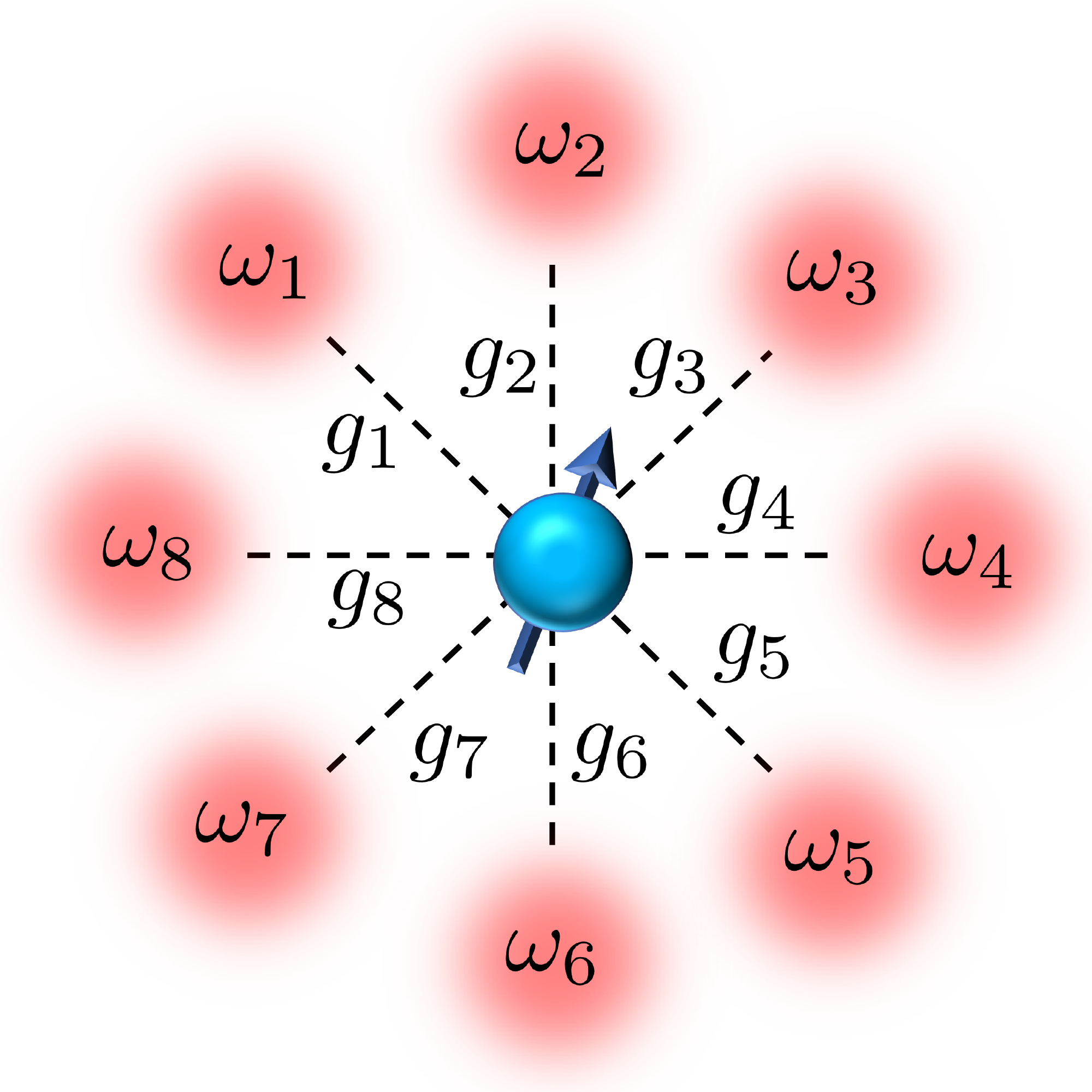}
    \caption{\textbf{Spin-boson model:} A two-level atom is interacting with multiple boson modes. Each boson mode is decoupled from each other, and has different angular frequency $\omega_k$. The atomic system is coupled to each mode with different interaction strength $g_k$. The Hamiltonian of the spin-boson model is described in Eq.~\eqref{eq:spin-boson-Hamiltonian}.}
    \label{fig:spin-boson}
\end{figure} 
This model can be descried by the following time-independent Hamiltonian (we set $\hbar=1$)
\begin{align}
    H=\frac{\omega_0}{2}\sigma_z +H_b+\sigma_z\otimes\sum_k(g_ka_k+g_k^* a_k\ad), 
\label{eq:spin-boson-Hamiltonian}
\end{align}
where 
\begin{align}
H_b \equiv \sum_k \omega_ka_k\ad a_k
\end{align}
is the bare Hamiltonian of the boson fields, and $\sigma_z\equiv\diag(1,-1)$ is the Pauli's Z operator acting on the atom. {$\omega_0$ and $\omega_k$ are the angular frequencies of the atom and the $k$-th boson mode, and $g_k$ denotes the coupling strength between the atom and the $k$-th boson mode.} Here, in general $g_k$ is a complex number, and $g_k^{*}$ denotes the complex conjugate of $g_k$. Also, $a_k(a_k\ad)$ is the annihilation (creation) operator of $k$-th mode of the boson fields. The Hamiltonian Eq.~\eqref{eq:spin-boson-Hamiltonian} describes an interaction between atom and boson fields with multiple modes, which leads to the dephasing process of the atomic system. In interaction picture, we have
\begin{align}
H(t)=\sigma_z\otimes\sum_k(g_k a_k e^{-i\omega_k t}+g_k^* a_k\ad e^{+i\omega_k t}).    
\end{align}
{During the evolution from $t=0$ to $t=\tau$, the internal energy change of the heat bath becomes (See Appendix.~\ref{app:InternalEnergyBath} for the proof)
\begin{align}
\begin{split}
\Delta E_b &\equiv \Tr\left[ U_{\tau} (\rho_{\text{in}}\otimes\rho_b^{\text{eq}}) U_{\tau}\ad H_b\right]-\Tr\left[\rho_b^{\text{eq}} H_b\right]\\
&=\sum_{k}\omega_k \abs{g_k}^2\left(\frac{\sin(\omega_k\tau/2)}{\omega_k/2}\right)^2\geq 0\,,
\end{split}
\label{eq:InternalEnergyBath}
\end{align}
which shows that the system undergoes the energy-emitting process, which can be verified from the energy conservation of the total system. Considering the noise spectral density $J(\omega)=\sum_k|g_k|^2\omega\delta(\omega-\omega_k)$, we can obtain
        \begin{align} 
        \Delta E_b=\int_{-\infty}^{\infty}J(\omega)\left(\frac{\sin(\omega\tau/2)}{\omega/2}\right)^2d\omega\,.
        \end{align}
        Since we have $\lim_{\tau\to\infty}\frac{\sin(\omega \tau/2)}{\omega/2}=\delta\left(\omega/2\right)=2\delta(\omega)$,
        where we used the relation $\lim_{\tau\to\infty} \left(\tau\frac{\sin(\omega\tau)}{\omega\tau}\right)=\delta(\omega)$, we obtain \begin{align}
        \lim_{\tau\to\infty}\Delta E_b=\int_{-\infty}^{\infty}4J(\omega)\delta^2(\omega)d\omega=4J(0)\delta(0)=0\,.
        \end{align}
	This is intuitively consistent because there will be no energy exchange between a small system and a large bath for the dephasing process at time $\tau\to \infty$.} 

{In the spin-boson model, the work reservoir is not taken into account. Therefore, we have $\Delta S \geq -\beta\Delta E_b$. Since the two-level atom undergoes the unital evolution, the total information production has to be $\Delta S\geq 0$. However, from Eq.~\eqref{eq:InternalEnergyBath}, we have $-\beta\Delta E_b\leq 0$, which shows that the lower bound in the maximum available work theorem for the spin-boson model is not tight enough. Instead, from Corollary.~\ref{cor:OTMUnitalMapBound}, we can find that $L_{\text{otm}}$ is the tighter bound as demonstrated in   Eq.~\eqref{eq:TighterBoundEnergyExchange}. Because $L_{\text{otm}}$ is characterized by the quantum cross entropy, the quantum cross entropy becomes a more meaningful quantity to characterize the quantum process of a system induced from its interaction with a heat bath.}

\section{Conclusion}
\label{sec:conc}

{In conclusion, we have proposed a distribution of an information production for a quantum state of arbitrary rank and a quantum channel by adopting the one-time measurement scheme. The derived Jarzynski-like equality and the lower bound on the total information production are characterized by the quantum cross entropy, which further enables one to explore the roles of quantum cross entropy in quantum communications, quantum machine learning and quantum thermodynamics. By focusing on the quantum autoencoder, we have explored the informational contributions of the quantum cross entropy of the compressed states to the loss of the maximum classical information transmittable through the circuit and the performance of the protocol characterized by the global cost function. We have also demonstrated the application of our result in the quantum thermodynamics by exploring relation between the quantum cross entropy and the maximum available work theorem. Our result can provide insights of the quantum cross entropy to assist designing quantum information processing protocols which utilize the quantum cross entropy as a resource to achieve some tasks. As a valuable future direction, we will explore reverse process in the OTM scheme, which still remains open. }

\section*{Acknowledgements}
We would like to thank {Sebastian Deffner}, Hailan Ma, Yuanlong Wang, and Marco Cerezo for helpful discussions. This work is supported by NSF: Proposal 2036347 (STTR Phase 1).
A.S. was supported by the internal R\&D from Aliro Technologies, Inc. Now, he is supported by the startup package from University of Massachusetts Boston. 
N.Y. is supported by MEXT Quantum Leap Flagship Program Grants No. JPMXS0118067285 and No. JPMXS0120319794.
T.J.H is supported by the graduate study program at University of Massachusetts Boston. 
Work by P.N. is supported by the NSF RAISE-QAC-QSA, Grant No. DMR-2037783 and the Department of Energy, Office of Basic Energy Sciences Grant DE-SC0019215. 
\\


	\setcounter{section}{0}
	\setcounter{figure}{0}
	\setcounter{corollary}{0}
	\setcounter{theorem}{0}
	\setcounter{proposition}{0}
	\setcounter{equation}{0}

	\appendix

\section*{Appendix}

\section{{Proof of Eq.~\eqref{eq:LandGlobalCost}}}
\label{app:LandGlobalCost}
	
Here, we provide a detailed proof of Eq.~\eqref{eq:LandGlobalCost}. First, by using the cost function $\mathscr{C}$, we can obtain an upper bound on $\Delta S$. From 
\begin{align}
    \rho_{\text{out}} = U\ad\left(\rho_A\otimes\dya{\psi}\right)U
\end{align}
with $\rho_A\equiv\Tr_B[U\rho_{\text{in}}U\ad]$,
because the von-Neumann entropy is unitarily invariant, we have 
\begin{align}
\begin{split}
    S(\rho_{\text{out}}) &= S(\rho_A\otimes\dya{\psi})=S(\rho_A)\\
    S(\rho_{\text{in}})&=S(U\rho_{\text{in}}U\ad)\,,
\end{split}
\end{align}
which leads to
\begin{align}
    \Delta S = S(\rho_A)-S(U\rho_{\text{in}}U\ad)\,.
\end{align}
Because $\eta_B\equiv\Tr_A\left[U\rho_{\text{in}}U\ad\right]$, from Araki-Lieb inequality~\cite{Araki70}
\begin{align}
    \abs{S(\rho_A)-S(\eta_B)}\leq S(U\rho_{\text{in}}U\ad)\,,
\end{align}
we have  
\begin{align}
\Delta S \leq S(\eta_B)\,.
\end{align}
By using $d_B$ the dimension of the Hilbert space $\mathcal{H}_B$, $S(\eta_B)$ can be upper bounded as
\begin{align}
\begin{split}
    S(\eta_B) 
    &=\ln(d_B)-S\left(\eta_B\,\Big|\Big|\,\frac{\id_B}{d_B}\right)\\
    &\leq\ln(d_B)-S_{\min}\left(\eta_B\,\Big|\Big|\,\frac{\id_B}{d_B}\right)\,,
\end{split}
\end{align}
where 
\begin{align}
    S_{\min}(\rho_1||\rho_2)\equiv -\ln\left(F[\rho_1,\rho_2]\right)
\end{align}
denotes the sandwiched min relative entropy of $\rho_1$ with respect to $\rho_2$~\cite{wilde2014strong,beigi2013sandwiched,Sagawa20} with the standard quantum fidelity defined as
\begin{align}
F[\rho_1,\rho_2]\equiv \left(\Tr\left[\sqrt{\rho_1^{1/2}\rho_2\rho_1^{1/2}}\right]\right)^2\,.
\end{align}
Therefore, we have 
\begin{align}
    S(\eta_B)\leq \ln\left(d_B F\left[\eta_B,\frac{\id_B}{d_B}\right]\right)\,.
\label{eq:EntropyUpperbound}
\end{align}

Here, we consider so-called generalized quantum fidelity~\cite{Tomamichel2010duality,Tomamichel2015quantum,Cappellini2007subnormalized,Cerezo19Fidelity}, which is defined as
\begin{align}
\widetilde{F}\left[\sigma_1,\sigma_2\right]\equiv \left(\sqrt{F\left[\sigma_1,\sigma_2\right]}+\sqrt{(1-\Tr\left[\sigma_1\right])(1-\Tr\left[\sigma_2\right])}\right)^2.
\end{align}
Here, note that $\sigma_1$ and $\sigma_2$ are the sub-normalized states i.e. $0\leq \Tr\left[\sigma_1\right]\leq 1$ and $0\leq \Tr[\sigma_2]\leq 1$. Because applying the projection operator $\dya{\psi}$ is described by the completely positive trace non-increasing (CPTNI) map~\cite{Cerezo19Fidelity}, from the monotonicity of the generalized quantum fidelity under the CPTNI maps~\cite{Tomamichel2010duality,Tomamichel2015quantum,Cappellini2007subnormalized,Cerezo19Fidelity}, we have 
\begin{align}
\begin{split}
    F\left[\eta_B,\frac{\id_B}{d_B}\right] 
    &=\widetilde{F}\left[\eta_B,\frac{\id_B}{d_B}\right]\\
    &\leq\widetilde{F}\left[\dya{\psi}\eta_B\dya{\psi},\frac{1}{d_B}\dya{\psi}\right]\,.
\end{split}
\end{align}
Since we have
\begin{align}
\begin{split}
\widetilde{F}&\left[\dya{\psi}\eta_B\dya{\psi},\frac{1}{d_B}\dya{\psi}\right]\\
&=\left(\sqrt{\frac{1-\mathscr{C}}{d_B}}+\sqrt{\mathscr{C}\left(1-\frac{1}{d_B}\right)}\right)^2\,,
\end{split}
\end{align}
we can obtain 
\begin{align}
    \Delta S \leq S(\eta_B)\leq 2\ln\left(\sqrt{1-\mathscr{C}}+\sqrt{(d_B-1)\mathscr{C}}\right)\,,
\end{align}
which states that the information production in quantum autoencoder with pure fresh-qubit state can be upper bounded by using $\mathscr{C}$. 
Therefore, from Theorem.~\ref{th:JarzynskiOTM}, we can finally arrive at 
Eq.~\eqref{eq:LandGlobalCost}
\begin{align}
    L_{\text{otm}} \leq \Delta S\leq 2\ln\left(\sqrt{1-\mathscr{C}}+\sqrt{(d_B-1)\mathscr{C}}\right)\,.
\end{align}

\section{{Proof of Eq.~\eqref{eq:InternalEnergyBath}}}
\label{app:InternalEnergyBath}
We demonstrate the detailed proof of Eq.~\eqref{eq:InternalEnergyBath} based on Refs.~\cite{cappellaro201222,Sone20a}. The Hamiltonian of the spin-boson model is
        \begin{align}
         H = \frac{\omega_0}{2}\sigma_z+ H_b+\sigma_z\otimes\sum_k(g_ka_k+g_k^*a_k\ad)\,,   
        \end{align}
        where we define $\sigma_z\equiv\begin{pmatrix}1&0\\0&-1\end{pmatrix}$ and $a_k(a_k\ad)$ as the annihilation (creation) operator of $k$-th mode of the boson heat bath. The annihilation and creation operators satisfy the commutation relation
        \begin{align}
            [a_k,a_{k'}\ad]=\delta_{kk'},~[a_k,a_{k'}]=[a_k\ad,a_{k'}\ad]=0\,.
        \end{align}
        Also, $H_b$ is defined as 
        \begin{align}
          H_b\equiv\sum_{k}\omega_ka_k\ad a_k\,.
        \end{align}
        Then, the Hamiltonian in the interaction picture becomes
        \begin{align}
            H(t) = \sigma_z\otimes \sum_{k}(g_ka_ke^{-i\omega_k t}+g_k^*a_k\ad e^{i\omega_k t})\,. 
        \end{align}
        Using Magnus expansion, the propagator becomes
\begin{align}
    U_t=\exp\left[-it (\overline{H}_0+\overline{H}_1)\right],
    \label{eq:spin-boson-propa}
\end{align}
where the higher terms are vanishing because $[H(t_1),H(t_2)]$ becomes just a number (See Eq.~\eqref{eq:CommutatorNumber}). Here, we define
\begin{align}
    \overline{H}_0\equiv\frac{1}{t}\int_{0}^{t}H(t_1)dt_1
\end{align}
and
\begin{align}
    \overline{H}_1\equiv-\frac{i}{2t}\int_{0}^tdt_1\int_{0}^{t_1}dt_2 [H(t_1),H(t_2)].
\end{align}
More explicitly, $\overline{H}_0$ can be written as 
\begin{align}
    \overline{H}_0= \sigma_z \otimes\sum_{k}\left(G_k(t)a_k+G^*(t)a_k\ad \right),
\label{eq:H_0}
\end{align}
where 
\begin{align}
     G_k(t)\equiv g_k\frac{\sin(\omega_kt/2)}{\omega_kt/2}e^{-i\omega_kt/2}.
\label{eq:G}
\end{align}
For $\overline{H}_1$, because we have
\begin{align}
[H(t_1),H(t_2)]=-2i\sum_k|g_k|^2\sin\left(\omega_k(t_1-t_2)\right)\,
\label{eq:CommutatorNumber}
\end{align}
and
\begin{align}
\begin{split}
\int_{0}^{t}dt_1\int_{0}^{t_1}&dt_2\sin\left(\omega_k(t_1-t_2)\right)\\
&=\frac{1}{\omega_k}\left(t-\frac{1}{\omega_k}\sin(\omega_k t)\right)\,,
\end{split}
\end{align}
we can write
\begin{align}
    \overline{H}_1
    =-\sum_k\frac{|g_k|^2}{\omega_k}\left(1-\frac{\sin(\omega_k t)}{\omega_k t}\right)\in\mathbb{R}\,,
\label{eq:H_1}
\end{align}
which is just a real number.
 
Therefore, from Eqs.~ \eqref{eq:spin-boson-propa}, \eqref{eq:H_0}, and \eqref{eq:H_1}, the propagator becomes
\begin{align}
    U_t=\exp\left[-it\sum_k\left( \sigma_z \otimes(G_k(t)a_k+G_k^*(t)a_k\ad )\right)\right]e^{-it \overline{H}_1}\,.
\label{eq:final_propagator}
\end{align}
With this propagator, due to
\begin{align}
\begin{split}
\!\!\!\! \left[\sum_{k'}\sigma_z\otimes(G_{k'}(t)a_{k'}+G_{k'}^*(t)a_{k'}\ad),a_k\right]&\!=\!-G_k^*(t) \sigma_z\\
\!\!\!\! \left[\sum_{k'}\sigma_z\otimes(G_{k'}(t)a_{k'}+G_{k'}^*(t)a_{k'}\ad),a_k\ad\right]&\!=\!G_k(t) \sigma_z\,,
\end{split}
\end{align}
from Baker-Hausdorff-Campbell's formula, we have
\begin{align}
\begin{split}
    U_t\ad a_k U_t&=a_k-it G_k^*(t) \sigma_z\\
    U_t\ad a_k\ad U_t&=a_k\ad+it G_k(t)\sigma_z \,.
\end{split}
\end{align}
Therefore, we can write
\begin{align}
\begin{split}
\!\!\!\!
U_t\ad H_bU_t=&\sum_{k}\omega_k \left(U_t\ad a_k\ad U_t\right)\left(U_t\ad a_kU_t\right)\\
    =&H_b+it\sum_k\omega_k  \sigma_z \otimes (G_k(t)a_k-G_k^*(t)a_k\ad)\\
    &+\sum_k\omega_k |G_k(t)|^2t^2\,.
\end{split}
\label{eq:HE_evolution}
\end{align}

Let $\rho_{\text{in}}$ be the input state of the system with a rank $r$, and the initial state of the boson heat bath be the Gibbs state
\begin{align}
    \rho_{b}^{\text{eq}} = \frac{e^{-\beta H_b}}{Z}\,.
\end{align}
Note that, with $a_k$ and $a_k\ad$, for all $k$, we have 
\begin{align}
    \Tr\left[\rho_b^{\text{eq}} a_k\right]=\Tr\left[\rho_b^{\text{eq}}a_k\ad\right] = 0\,.
\label{eq:Ave_a_adagger}
\end{align}
We assume that the two-level atomic system and bosonic field are initially decoupled. Therefore, the initial state of the total system is $\rho_{\text{in}}\otimes\rho_b^{\text{eq}}$, so that the evolution of the atomic system from $t=0$ to $t=\tau$ is described by the following thermal operation
\begin{align}
    \rho_{\text{out}} = \Phi(\rho_{\text{in}}) = \Tr_b\left[U_{\tau} (\rho_{\text{in}}\otimes\rho_b^{\text{eq}}) U_{\tau}\ad\right]\,.
\end{align} 
The internal energy change of the heat bath during the evolution can be defined as the difference in the average energy of the heat bath at $t=\tau$ and $t=0$
\begin{align}
\Delta E_b \equiv \Tr\left[ U_{\tau} (\rho_{\text{in}}\otimes\rho_b^{\text{eq}}) U_{\tau}\ad H_b\right]-\Tr\left[\rho_b^{\text{eq}} H_b\right]\,. 
\end{align}
From Eqs.~\eqref{eq:G}, \eqref{eq:HE_evolution} and \eqref{eq:Ave_a_adagger}, $\Delta E_b$ can be explicitly written as
\begin{align}
    \Delta E_b = \sum_k\omega_k\abs{g_k}^2\left(\frac{\sin\left(\omega_k\tau/2\right)}{\omega_k/2}\right)^2\geq 0\,.
\end{align}

\section{{Second-law-like Inequality Involving Guessed Heat}}
\label{app:2ndLawGuessedHeat}
The information production distribution $\widetilde{P}(\sigma)$ can be related to the distribution of the internal energy difference in OTM scheme $\widetilde{P}(\Delta E_s)$ in a very special case, which leads to a second-law-like inequality involving the guessed heat introduced in Ref.~\cite{Sone20a}. Let $H_s(t)$ be the system's bare Hamiltonian, which is time-dependent. Also, suppose that the system is initially decoupled from the heat bath, which is initially prepared in a Gibbs state
\begin{align}
    \rho_b^{\text{eq}}=\frac{e^{-\beta H_b}}{Z_b}\,,
\end{align}
where $H_b$ is the bath's bare Hamiltonian, which is time-independent. Here, $Z_b\equiv\Tr\left[e^{-\beta H_b}\right]$ is the partition function defined by $H_b$. Let $H_{\text{int}}$ be the interaction Hamiltonian. Then, the unitary operator $U_{t}$ describing the time evolution of the total system follows the Schr\"{o}dinger's equation $\partial_{t}U_{t} = -i(H_s(t)+H_b+H_{\text{int}})U_{t}$ with $U_0\equiv\id$. Evolving the total system from $t=0$ to $t=\tau$ and focusing on the system alone, we have  
\begin{align}
    \rho_{\text{out}}=\Phi(\rho_{\text{in}}) = \Tr_b\left[U_{\tau}(\rho_{\text{in}}\otimes \rho_{b}^{\text{eq}})U_{\tau}\ad\right]\,.
\end{align}

Let $\{E_i,\ket{E_i}\}_{i=1}^{d}$ be an eigensystem of $H_s(0)$. When we have 
\begin{align}
\begin{split}
\rho_{\text{in}}&=\rho_s^{\text{eq}}(0)\equiv \frac{e^{-\beta H_s(0)}}{Z_0}\\
\rho_{\text{out}}&=\rho_s^{\text{eq}}(\tau)\equiv \frac{e^{-\beta H_s(\tau)}}{Z_{\tau}}\,,
\end{split}
\end{align}
where $Z_0\equiv\Tr\left[e^{-\beta H_s(0)}\right]$ and $Z_{\tau}\equiv \Tr\left[e^{-\beta H_s(\tau)}\right]$ are the partition functions defined by $H_s(0)$ and $H_s(\tau)$, respectively, from Eq.~\eqref{eq:DisOTM}, the information production distribution in the OTM scheme becomes
        \begin{align}
          \widetilde{P}(\sigma)=\frac{1}{\beta}\sum_{i=1}^{d}\frac{e^{-\beta E_i}}{Z_0}\delta\left(\frac{\sigma}{\beta}+\Delta F-\Delta\widetilde{E}(E_i)\right)\,,
        \end{align}
        where 
        \begin{align}
        \Delta\widetilde{E}(E_i)\equiv \Tr\left[\Phi(\dya{E_i})H_{s}(\tau)\right]-E_i\,,
        \end{align}
        and 
        \begin{align}
        \Delta F \equiv -\beta^{-1}\ln\left(\frac{Z_{\tau}}{Z_0}\right)
        \end{align}
        is the equilibrium Helmholtz free energy difference.  
        From Ref.~\cite{Sone20a}, $\widetilde{P}(\Delta E_s)$ is given by 
        \begin{align}
            \widetilde{P}(\Delta E_s) \!=\! \sum_{i=1}^{d}\frac{e^{-\beta E_i}}{Z_0}\delta\left(\Delta E_s-\Delta\widetilde{E}(E_i)\right)\,.
        \end{align}
        Therefore, we can write
        \begin{align}
            \widetilde{P}(\Delta E_s)=\beta\widetilde{P}(\sigma)
            \label{eq:app:EandInfProd}
        \end{align}
        with the random variable $\sigma$ being
        \begin{align}
            \sigma = \beta(\Delta E_s-\Delta F)\,.
            \label{eq:app:SigmaEandF}
        \end{align}
        Following Ref.~\cite{Sone20a}, we have 
        \begin{align}
            \ave{e^{-\beta \Delta E_s}}_{\widetilde{P}}=e^{-\beta \Delta F}e^{-\beta \ave{\widetilde{Q}}_b}e^{-S(\Theta_{sb}(\tau)||\rho_s^{\text{eq}}(\tau)\otimes\rho_{b}^{\text{eq}})}\,,
        \label{eq:app:OTMJarzynski}
        \end{align}
        where
       \begin{align}
            \Theta_{sb}(\tau) \equiv \sum_{i=1}^{d}\frac{e^{-\beta\Tr\left[\Phi(\dya{E_i})H_s(\tau)\right]}}{\widetilde{Z}_{\tau}}U_{\tau}\left(\dya{E_i}\otimes\rho_b^{\text{eq}}\right)U_{\tau}\ad\,
        \end{align}
        is called ``guessed state", and $\ave{\widetilde{Q}}_b$ is a heat-like quantity called ``guessed heat" defined as 
        \begin{align}
            \ave{\widetilde{Q}}_b \equiv \Tr\left[H_b\rho_b^{\text{eq}}\right]-\Tr\left[H_b\Theta_{sb}(\tau)\right]\,.
        \end{align}
        This heat-like quantity describes an energy dissipation from the heat bath as if its final state is $\Tr_s\left[\Theta_{sb}(\tau)\right]$ the reduced state of the guessed state. From Eq.~\eqref{eq:app:EandInfProd} and \eqref{eq:app:SigmaEandF}, we can obtain
        \begin{align}
            \ave{e^{-\beta \Delta E_s}}_{\widetilde{P}} = e^{-\beta \Delta F}\ave{e^{-\sigma}}_{\widetilde{P}}\,.
        \end{align}
        From Eq.~ \eqref{eq:app:OTMJarzynski}, we can finally write
        \begin{align}
        \ave{e^{-\sigma}}_{\widetilde{P}}=e^{-\beta\ave{\widetilde{Q}}_b}e^{-S(\Theta_{sb}(\tau)||\rho_s^{\text{eq}}(\tau)\otimes\rho_b^{\text{eq}})}\,.
        \end{align}
        By using Jensen's inequality and the non-negativity of the quantum relative entropy, we can arrive at 
        \begin{align}
        \Delta S-\beta\ave{\widetilde{Q}}_b\geq 0\,,
        \end{align}
        which is the second-law-like inequality involving the guessed heat.

\bibliography{ref.bib}

\begin{thebibliography}{91}%
\makeatletter
\providecommand \@ifxundefined [1]{%
 \@ifx{#1\undefined}
}%
\providecommand \@ifnum [1]{%
 \ifnum #1\expandafter \@firstoftwo
 \else \expandafter \@secondoftwo
 \fi
}%
\providecommand \@ifx [1]{%
 \ifx #1\expandafter \@firstoftwo
 \else \expandafter \@secondoftwo
 \fi
}%
\providecommand \natexlab [1]{#1}%
\providecommand \enquote  [1]{``#1''}%
\providecommand \bibnamefont  [1]{#1}%
\providecommand \bibfnamefont [1]{#1}%
\providecommand \citenamefont [1]{#1}%
\providecommand \href@noop [0]{\@secondoftwo}%
\providecommand \href [0]{\begingroup \@sanitize@url \@href}%
\providecommand \@href[1]{\@@startlink{#1}\@@href}%
\providecommand \@@href[1]{\endgroup#1\@@endlink}%
\providecommand \@sanitize@url [0]{\catcode `\\12\catcode `\$12\catcode
  `\&12\catcode `\#12\catcode `\^12\catcode `\_12\catcode `\%12\relax}%
\providecommand \@@startlink[1]{}%
\providecommand \@@endlink[0]{}%
\providecommand \url  [0]{\begingroup\@sanitize@url \@url }%
\providecommand \@url [1]{\endgroup\@href {#1}{\urlprefix }}%
\providecommand \urlprefix  [0]{URL }%
\providecommand \Eprint [0]{\href }%
\providecommand \doibase [0]{https://doi.org/}%
\providecommand \selectlanguage [0]{\@gobble}%
\providecommand \bibinfo  [0]{\@secondoftwo}%
\providecommand \bibfield  [0]{\@secondoftwo}%
\providecommand \translation [1]{[#1]}%
\providecommand \BibitemOpen [0]{}%
\providecommand \bibitemStop [0]{}%
\providecommand \bibitemNoStop [0]{.\EOS\space}%
\providecommand \EOS [0]{\spacefactor3000\relax}%
\providecommand \BibitemShut  [1]{\csname bibitem#1\endcsname}%
\let\auto@bib@innerbib\@empty
\bibitem [{\citenamefont {Deffner}\ and\ \citenamefont
  {Campbell}(2019)}]{DeffnerBook19}%
  \BibitemOpen
  \bibfield  {author} {\bibinfo {author} {\bibfnamefont {S.}~\bibnamefont
  {Deffner}}\ and\ \bibinfo {author} {\bibfnamefont {S.}~\bibnamefont
  {Campbell}},\ }\href {https://iopscience.iop.org/book/978-1-64327-658-8}
  {\emph {\bibinfo {title} {Quantum Thermodynamics}}}\ (\bibinfo  {publisher}
  {Morgan and Claypool Publishers, San Rafael},\ \bibinfo {year}
  {2019})\BibitemShut {NoStop}%
\bibitem [{\citenamefont {Binder}\ \emph {et~al.}(2019)\citenamefont {Binder},
  \citenamefont {Correa}, \citenamefont {Gogolin}, \citenamefont {Anders},\
  and\ \citenamefont {Adesso}}]{Binder19}%
  \BibitemOpen
  \bibfield  {author} {\bibinfo {author} {\bibfnamefont {F.}~\bibnamefont
  {Binder}}, \bibinfo {author} {\bibfnamefont {L.~A.}\ \bibnamefont {Correa}},
  \bibinfo {author} {\bibfnamefont {C.}~\bibnamefont {Gogolin}}, \bibinfo
  {author} {\bibfnamefont {J.}~\bibnamefont {Anders}},\ and\ \bibinfo {author}
  {\bibfnamefont {G.}~\bibnamefont {Adesso}},\ }\href
  {https://link.springer.com/book/10.1007/978-3-319-99046-0} {\emph {\bibinfo
  {title} {Thermodynamics in the Quantum Regime}}}\ (\bibinfo  {publisher}
  {Springer},\ \bibinfo {year} {2019})\BibitemShut {NoStop}%
\bibitem [{\citenamefont {Vinjanampathy}\ and\ \citenamefont
  {Anders}(2016)}]{Anders16}%
  \BibitemOpen
  \bibfield  {author} {\bibinfo {author} {\bibfnamefont {S.}~\bibnamefont
  {Vinjanampathy}}\ and\ \bibinfo {author} {\bibfnamefont {J.}~\bibnamefont
  {Anders}},\ }\bibfield  {title} {\bibinfo {title} {Quantum thermodynamics},\
  }\href {https://doi.org/10.1080/00107514.2016.1201896} {\bibfield  {journal}
  {\bibinfo  {journal} {Contemp. Phys.}\ }\textbf {\bibinfo {volume} {57}},\
  \bibinfo {pages} {545} (\bibinfo {year} {2016})}\BibitemShut {NoStop}%
\bibitem [{\citenamefont {Sagawa}(2012)}]{sagawa2012thermodynamics}%
  \BibitemOpen
  \bibfield  {author} {\bibinfo {author} {\bibfnamefont {T.}~\bibnamefont
  {Sagawa}},\ }\bibfield  {title} {\bibinfo {title} {Thermodynamics of
  information processing in small systems},\ }\href
  {https://doi.org/10.1143/PTP.127.1} {\bibfield  {journal} {\bibinfo
  {journal} {Prog. Theor. Phys.}\ }\textbf {\bibinfo {volume} {127}},\ \bibinfo
  {pages} {1} (\bibinfo {year} {2012})}\BibitemShut {NoStop}%
\bibitem [{\citenamefont {Goold}\ \emph {et~al.}(2016)\citenamefont {Goold},
  \citenamefont {Huber}, \citenamefont {Riera}, \citenamefont {del Rio},\ and\
  \citenamefont {Skrzypczyk}}]{Goold16}%
  \BibitemOpen
  \bibfield  {author} {\bibinfo {author} {\bibfnamefont {J.}~\bibnamefont
  {Goold}}, \bibinfo {author} {\bibfnamefont {M.}~\bibnamefont {Huber}},
  \bibinfo {author} {\bibfnamefont {A.}~\bibnamefont {Riera}}, \bibinfo
  {author} {\bibfnamefont {L.}~\bibnamefont {del Rio}},\ and\ \bibinfo {author}
  {\bibfnamefont {P.}~\bibnamefont {Skrzypczyk}},\ }\bibfield  {title}
  {\bibinfo {title} {The role of quantum information in thermodynamics—a
  topical review},\ }\href
  {https://iopscience.iop.org/article/10.1088/1751-8113/49/14/143001/meta}
  {\bibfield  {journal} {\bibinfo  {journal} {J. Phys. A: Math. Theor.}\
  }\textbf {\bibinfo {volume} {49}},\ \bibinfo {pages} {143001} (\bibinfo
  {year} {2016})}\BibitemShut {NoStop}%
\bibitem [{\citenamefont {Ng}\ and\ \citenamefont
  {Woods}(2018)}]{ng2018resource}%
  \BibitemOpen
  \bibfield  {author} {\bibinfo {author} {\bibfnamefont {N.~H.~Y.}\
  \bibnamefont {Ng}}\ and\ \bibinfo {author} {\bibfnamefont {M.~P.}\
  \bibnamefont {Woods}},\ }\bibfield  {title} {\bibinfo {title} {Resource
  theory of quantum thermodynamics: Thermal operations and second laws},\ }in\
  \href {https://doi.org/10.1007/978-3-319-99046-0_26} {\emph {\bibinfo
  {booktitle} {Thermodynamics in the Quantum Regime}}}\ (\bibinfo  {publisher}
  {Springer},\ \bibinfo {year} {2018})\ pp.\ \bibinfo {pages}
  {625--650}\BibitemShut {NoStop}%
\bibitem [{\citenamefont {Auff{\'{e}}ves}(2022)}]{Auffeves2022energy}%
  \BibitemOpen
  \bibfield  {author} {\bibinfo {author} {\bibfnamefont {A.}~\bibnamefont
  {Auff{\'{e}}ves}},\ }\bibfield  {title} {\bibinfo {title} {Quantum
  {T}echnologies {N}eed a {Q}uantum {E}nergy {I}nitiative},\ }\href
  {https://doi.org/10.1103/PRXQuantum.3.020101} {\bibfield  {journal} {\bibinfo
   {journal} {PRX Quantum}\ }\textbf {\bibinfo {volume} {3}},\ \bibinfo {pages}
  {020101} (\bibinfo {year} {2022})}\BibitemShut {NoStop}%
\bibitem [{\citenamefont {Deffner}(2021)}]{deffner2021energetic}%
  \BibitemOpen
  \bibfield  {author} {\bibinfo {author} {\bibfnamefont {S.}~\bibnamefont
  {Deffner}},\ }\bibfield  {title} {\bibinfo {title} {Energetic cost of
  hamiltonian quantum gates},\ }\href
  {https://doi.org/10.1209/0295-5075/134/40002} {\bibfield  {journal} {\bibinfo
   {journal} {EPL}\ }\textbf {\bibinfo {volume} {134}},\ \bibinfo {pages}
  {40002} (\bibinfo {year} {2021})}\BibitemShut {NoStop}%
\bibitem [{\citenamefont {Aifer}\ and\ \citenamefont
  {Deffner}(2022)}]{aifer2022quantum}%
  \BibitemOpen
  \bibfield  {author} {\bibinfo {author} {\bibfnamefont {M.}~\bibnamefont
  {Aifer}}\ and\ \bibinfo {author} {\bibfnamefont {S.}~\bibnamefont
  {Deffner}},\ }\bibfield  {title} {\bibinfo {title} {From quantum speed limits
  to energy-efficient quantum gates},\ }\href
  {https://doi.org/10.1088/1367-2630/ac6821} {\bibfield  {journal} {\bibinfo
  {journal} {New J. Phys.}\ }\textbf {\bibinfo {volume} {24}},\ \bibinfo
  {pages} {055002} (\bibinfo {year} {2022})}\BibitemShut {NoStop}%
\bibitem [{\citenamefont {Buffoni}\ \emph {et~al.}(2022)\citenamefont
  {Buffoni}, \citenamefont {Gherardini}, \citenamefont {Cruzeiro},\ and\
  \citenamefont {Omar}}]{buffoni2022third}%
  \BibitemOpen
  \bibfield  {author} {\bibinfo {author} {\bibfnamefont {L.}~\bibnamefont
  {Buffoni}}, \bibinfo {author} {\bibfnamefont {S.}~\bibnamefont {Gherardini}},
  \bibinfo {author} {\bibfnamefont {E.~Z.}\ \bibnamefont {Cruzeiro}},\ and\
  \bibinfo {author} {\bibfnamefont {Y.}~\bibnamefont {Omar}},\ }\bibfield
  {title} {\bibinfo {title} {Third law of thermodynamics and the scaling of
  quantum computers},\ }\href {https://arxiv.org/abs/2203.09545} {\bibfield
  {journal} {\bibinfo  {journal} {arXiv:2203.09545}\ } (\bibinfo {year}
  {2022})}\BibitemShut {NoStop}%
\bibitem [{\citenamefont {Evans}\ \emph {et~al.}(1993)\citenamefont {Evans},
  \citenamefont {Cohen},\ and\ \citenamefont {Morriss}}]{evans1993probability}%
  \BibitemOpen
  \bibfield  {author} {\bibinfo {author} {\bibfnamefont {D.~J.}\ \bibnamefont
  {Evans}}, \bibinfo {author} {\bibfnamefont {E.~G.~D.}\ \bibnamefont
  {Cohen}},\ and\ \bibinfo {author} {\bibfnamefont {G.~P.}\ \bibnamefont
  {Morriss}},\ }\bibfield  {title} {\bibinfo {title} {Probability of second law
  violations in shearing steady states},\ }\href
  {https://doi.org/10.1103/PhysRevLett.71.2401} {\bibfield  {journal} {\bibinfo
   {journal} {Phys. Rev. Lett.}\ }\textbf {\bibinfo {volume} {71}},\ \bibinfo
  {pages} {2401} (\bibinfo {year} {1993})}\BibitemShut {NoStop}%
\bibitem [{\citenamefont {Jarzynski}(1997)}]{Jarzynski97}%
  \BibitemOpen
  \bibfield  {author} {\bibinfo {author} {\bibfnamefont {C.}~\bibnamefont
  {Jarzynski}},\ }\bibfield  {title} {\bibinfo {title} {Nonequilibrium equality
  for free energy differences},\ }\href
  {https://doi.org/10.1103/PhysRevLett.78.2690} {\bibfield  {journal} {\bibinfo
   {journal} {Phys. Rev. Lett.}\ }\textbf {\bibinfo {volume} {78}},\ \bibinfo
  {pages} {2690} (\bibinfo {year} {1997})}\BibitemShut {NoStop}%
\bibitem [{\citenamefont {Crooks}(1999)}]{crooks1999entropy}%
  \BibitemOpen
  \bibfield  {author} {\bibinfo {author} {\bibfnamefont {G.~E.}\ \bibnamefont
  {Crooks}},\ }\bibfield  {title} {\bibinfo {title} {Entropy production
  fluctuation theorem and the nonequilibrium work relation for free energy
  differences},\ }\href {https://doi.org/10.1103/PhysRevE.60.2721} {\bibfield
  {journal} {\bibinfo  {journal} {Phys. Rev. E}\ }\textbf {\bibinfo {volume}
  {60}},\ \bibinfo {pages} {2721} (\bibinfo {year} {1999})}\BibitemShut
  {NoStop}%
\bibitem [{\citenamefont {Hatano}\ and\ \citenamefont
  {Sasa}(2001)}]{hatano2001steady}%
  \BibitemOpen
  \bibfield  {author} {\bibinfo {author} {\bibfnamefont {T.}~\bibnamefont
  {Hatano}}\ and\ \bibinfo {author} {\bibfnamefont {S.-i.}\ \bibnamefont
  {Sasa}},\ }\bibfield  {title} {\bibinfo {title} {Steady-state
  {T}hermodynamics of {L}angevin {S}ystems},\ }\href
  {https://doi.org/10.1103/PhysRevLett.86.3463} {\bibfield  {journal} {\bibinfo
   {journal} {Phys. Rev. Lett.}\ }\textbf {\bibinfo {volume} {86}},\ \bibinfo
  {pages} {3463} (\bibinfo {year} {2001})}\BibitemShut {NoStop}%
\bibitem [{\citenamefont {De~Chiara}\ and\ \citenamefont
  {Imparato}(2022)}]{de2022quantum}%
  \BibitemOpen
  \bibfield  {author} {\bibinfo {author} {\bibfnamefont {G.}~\bibnamefont
  {De~Chiara}}\ and\ \bibinfo {author} {\bibfnamefont {A.}~\bibnamefont
  {Imparato}},\ }\bibfield  {title} {\bibinfo {title} {Quantum fluctuation
  theorem for dissipative processes},\ }\href
  {https://doi.org/10.1103/PhysRevResearch.4.023230} {\bibfield  {journal}
  {\bibinfo  {journal} {Phys. Rev. Research}\ }\textbf {\bibinfo {volume}
  {4}},\ \bibinfo {pages} {023230} (\bibinfo {year} {2022})}\BibitemShut
  {NoStop}%
\bibitem [{\citenamefont {de~Z{\'{a}}rate}(2011)}]{Ortiz2011}%
  \BibitemOpen
  \bibfield  {author} {\bibinfo {author} {\bibfnamefont {J.~M.~O.}\
  \bibnamefont {de~Z{\'{a}}rate}},\ }\bibfield  {title} {\bibinfo {title}
  {Interview with michael e. fisher},\ }\href
  {https://www.europhysicsnews.org/articles/epn/abs/2011/01/epn2011421p14/epn2011421p14.html}
  {\bibfield  {journal} {\bibinfo  {journal} {Europhysics News}\ }\textbf
  {\bibinfo {volume} {42}},\ \bibinfo {pages} {14} (\bibinfo {year}
  {2011})}\BibitemShut {NoStop}%
\bibitem [{\citenamefont {Jarzynski}(2013)}]{jarzynski2013equalities}%
  \BibitemOpen
  \bibfield  {author} {\bibinfo {author} {\bibfnamefont {C.}~\bibnamefont
  {Jarzynski}},\ }\bibfield  {title} {\bibinfo {title} {Equalities and
  inequalities: Irreversibility and the second law of thermodynamics at the
  nanoscale},\ }\href
  {https://doi.org/10.1146/annurev-conmatphys-062910-140506} {\bibfield
  {journal} {\bibinfo  {journal} {Annu. Rev. Condens. Matter Phys.}\ ,\
  \bibinfo {pages} {145}} (\bibinfo {year} {2013})}\BibitemShut {NoStop}%
\bibitem [{\citenamefont {Andrieux}\ and\ \citenamefont
  {Gaspard}(2008)}]{andrieux2008quantum}%
  \BibitemOpen
  \bibfield  {author} {\bibinfo {author} {\bibfnamefont {D.}~\bibnamefont
  {Andrieux}}\ and\ \bibinfo {author} {\bibfnamefont {P.}~\bibnamefont
  {Gaspard}},\ }\bibfield  {title} {\bibinfo {title} {Quantum work relations
  and response theory},\ }\href
  {https://doi.org/10.1103/PhysRevLett.100.230404} {\bibfield  {journal}
  {\bibinfo  {journal} {Phys. Rev. Lett.}\ }\textbf {\bibinfo {volume} {100}},\
  \bibinfo {pages} {230404} (\bibinfo {year} {2008})}\BibitemShut {NoStop}%
\bibitem [{\citenamefont {Andrieux}\ \emph {et~al.}(2009)\citenamefont
  {Andrieux}, \citenamefont {Gaspard}, \citenamefont {Monnai},\ and\
  \citenamefont {Tasaki}}]{andrieux2009fluctuation}%
  \BibitemOpen
  \bibfield  {author} {\bibinfo {author} {\bibfnamefont {D.}~\bibnamefont
  {Andrieux}}, \bibinfo {author} {\bibfnamefont {P.}~\bibnamefont {Gaspard}},
  \bibinfo {author} {\bibfnamefont {T.}~\bibnamefont {Monnai}},\ and\ \bibinfo
  {author} {\bibfnamefont {S.}~\bibnamefont {Tasaki}},\ }\bibfield  {title}
  {\bibinfo {title} {The fluctuation theorem for currents in open quantum
  systems},\ }\href {https://doi.org/10.1088/1367-2630/11/4/043014} {\bibfield
  {journal} {\bibinfo  {journal} {New J. Phys.}\ }\textbf {\bibinfo {volume}
  {11}},\ \bibinfo {pages} {043014} (\bibinfo {year} {2009})}\BibitemShut
  {NoStop}%
\bibitem [{\citenamefont {Sagawa}\ and\ \citenamefont {Ueda}(2010)}]{Sagawa10}%
  \BibitemOpen
  \bibfield  {author} {\bibinfo {author} {\bibfnamefont {T.}~\bibnamefont
  {Sagawa}}\ and\ \bibinfo {author} {\bibfnamefont {M.}~\bibnamefont {Ueda}},\
  }\bibfield  {title} {\bibinfo {title} {Generalized {J}arzynski equality under
  nonequilibrium feedback control},\ }\href
  {https://doi.org/10.1103/PhysRevLett.104.090602} {\bibfield  {journal}
  {\bibinfo  {journal} {Phys. Rev. Lett.}\ }\textbf {\bibinfo {volume} {104}},\
  \bibinfo {pages} {090602} (\bibinfo {year} {2010})}\BibitemShut {NoStop}%
\bibitem [{\citenamefont {Fujitani}\ and\ \citenamefont
  {Suzuki}(2010)}]{fujitani2010jarzynski}%
  \BibitemOpen
  \bibfield  {author} {\bibinfo {author} {\bibfnamefont {Y.}~\bibnamefont
  {Fujitani}}\ and\ \bibinfo {author} {\bibfnamefont {H.}~\bibnamefont
  {Suzuki}},\ }\bibfield  {title} {\bibinfo {title} {Jarzynski equality
  modified in the linear feedback system},\ }\href
  {https://doi.org/10.1143/JPSJ.79.104003} {\bibfield  {journal} {\bibinfo
  {journal} {J. Phys. Soc. Jpn.}\ }\textbf {\bibinfo {volume} {79}},\ \bibinfo
  {pages} {104003} (\bibinfo {year} {2010})}\BibitemShut {NoStop}%
\bibitem [{\citenamefont {Deffner}\ \emph {et~al.}(2016)\citenamefont
  {Deffner}, \citenamefont {Paz},\ and\ \citenamefont {Zurek}}]{Deffner16}%
  \BibitemOpen
  \bibfield  {author} {\bibinfo {author} {\bibfnamefont {S.}~\bibnamefont
  {Deffner}}, \bibinfo {author} {\bibfnamefont {J.~P.}\ \bibnamefont {Paz}},\
  and\ \bibinfo {author} {\bibfnamefont {W.~H.}\ \bibnamefont {Zurek}},\
  }\bibfield  {title} {\bibinfo {title} {Quantum work and the thermodynamic
  cost of quantum measurements},\ }\href
  {https://journals.aps.org/pre/abstract/10.1103/PhysRevE.94.010103} {\bibfield
   {journal} {\bibinfo  {journal} {Phys. Rev. E}\ }\textbf {\bibinfo {volume}
  {94}},\ \bibinfo {pages} {010103(R)} (\bibinfo {year} {2016})}\BibitemShut
  {NoStop}%
\bibitem [{\citenamefont {Sone}\ and\ \citenamefont {Deffner}(2021)}]{Sone21b}%
  \BibitemOpen
  \bibfield  {author} {\bibinfo {author} {\bibfnamefont {A.}~\bibnamefont
  {Sone}}\ and\ \bibinfo {author} {\bibfnamefont {S.}~\bibnamefont {Deffner}},\
  }\bibfield  {title} {\bibinfo {title} {{J}arzynski equality for stochastic
  conditional work},\ }\href {https://doi.org/10.1007/s10955-021-02720-6}
  {\bibfield  {journal} {\bibinfo  {journal} {J. Stat. Phys}\ }\textbf
  {\bibinfo {volume} {183}},\ \bibinfo {pages} {11} (\bibinfo {year}
  {2021})}\BibitemShut {NoStop}%
\bibitem [{\citenamefont {Gardas}\ and\ \citenamefont
  {Deffner}(2018)}]{gardas2018quantum}%
  \BibitemOpen
  \bibfield  {author} {\bibinfo {author} {\bibfnamefont {B.}~\bibnamefont
  {Gardas}}\ and\ \bibinfo {author} {\bibfnamefont {S.}~\bibnamefont
  {Deffner}},\ }\bibfield  {title} {\bibinfo {title} {Quantum fluctuation
  theorem for error diagnostics in quantum annealers},\ }\href
  {https://doi.org/10.1038/s41598-018-35264-z} {\bibfield  {journal} {\bibinfo
  {journal} {Sci. Rep.}\ }\textbf {\bibinfo {volume} {8}},\ \bibinfo {pages}
  {1} (\bibinfo {year} {2018})}\BibitemShut {NoStop}%
\bibitem [{\citenamefont {Kafri}\ and\ \citenamefont
  {Deffner}(2012)}]{Kafri12}%
  \BibitemOpen
  \bibfield  {author} {\bibinfo {author} {\bibfnamefont {D.}~\bibnamefont
  {Kafri}}\ and\ \bibinfo {author} {\bibfnamefont {S.}~\bibnamefont
  {Deffner}},\ }\bibfield  {title} {\bibinfo {title} {Holevo's bound from a
  general quantum fluctuation theorem},\ }\href
  {https://journals.aps.org/pra/abstract/10.1103/PhysRevA.86.044302} {\bibfield
   {journal} {\bibinfo  {journal} {Phys. Rev. A}\ }\textbf {\bibinfo {volume}
  {86}},\ \bibinfo {pages} {044302} (\bibinfo {year} {2012})}\BibitemShut
  {NoStop}%
\bibitem [{\citenamefont {Holevo}(1973)}]{holevo1973bounds}%
  \BibitemOpen
  \bibfield  {author} {\bibinfo {author} {\bibfnamefont {A.~S.}\ \bibnamefont
  {Holevo}},\ }\bibfield  {title} {\bibinfo {title} {Bounds for the quantity of
  information transmitted by a quantum communication channel},\ }\href
  {http://www.mathnet.ru/php/archive.phtml?wshow=paper&jrnid=ppi&paperid=903&option_lang=eng}
  {\bibfield  {journal} {\bibinfo  {journal} {Probl. Pereda. Inf.}\ }\textbf
  {\bibinfo {volume} {9}},\ \bibinfo {pages} {3} (\bibinfo {year}
  {1973})}\BibitemShut {NoStop}%
\bibitem [{\citenamefont {Holevo}(1998)}]{holevo1998capacity}%
  \BibitemOpen
  \bibfield  {author} {\bibinfo {author} {\bibfnamefont {A.~S.}\ \bibnamefont
  {Holevo}},\ }\bibfield  {title} {\bibinfo {title} {The capacity of the
  quantum channel with general signal states},\ }\href
  {https://doi.org/10.1109/18.651037} {\bibfield  {journal} {\bibinfo
  {journal} {IEEE Trans. Info. Theo.}\ }\textbf {\bibinfo {volume} {44}},\
  \bibinfo {pages} {269} (\bibinfo {year} {1998})}\BibitemShut {NoStop}%
\bibitem [{\citenamefont {Watrous}(2018)}]{WatrousBook18}%
  \BibitemOpen
  \bibfield  {author} {\bibinfo {author} {\bibfnamefont {J.}~\bibnamefont
  {Watrous}},\ }\href@noop {} {\emph {\bibinfo {title} {The Theory of Quantum
  Information}}}\ (\bibinfo  {publisher} {Cambridge University Press},\
  \bibinfo {year} {2018})\BibitemShut {NoStop}%
\bibitem [{\citenamefont {Tasaki}(2000)}]{Tasaki00}%
  \BibitemOpen
  \bibfield  {author} {\bibinfo {author} {\bibfnamefont {H.}~\bibnamefont
  {Tasaki}},\ }\bibfield  {title} {\bibinfo {title} {Jarzynski relations for
  quantum systems and some applications},\ }\href
  {https://arxiv.org/abs/cond-mat/0009244v2} {\bibfield  {journal} {\bibinfo
  {journal} {arXiv:cond-mat/0009244}\ } (\bibinfo {year} {2000})}\BibitemShut
  {NoStop}%
\bibitem [{\citenamefont {Kurchan}(2000)}]{Kurchan01}%
  \BibitemOpen
  \bibfield  {author} {\bibinfo {author} {\bibfnamefont {J.}~\bibnamefont
  {Kurchan}},\ }\bibfield  {title} {\bibinfo {title} {A quantum fluctuation
  theorem},\ }\href {https://arxiv.org/abs/cond-mat/0007360v2} {\bibfield
  {journal} {\bibinfo  {journal} {arXiv:cond-mat/0007360}\ } (\bibinfo {year}
  {2000})}\BibitemShut {NoStop}%
\bibitem [{\citenamefont {Smith}\ \emph {et~al.}(2018)\citenamefont {Smith},
  \citenamefont {Lu}, \citenamefont {An}, \citenamefont {Zhang}, \citenamefont
  {Zhang}, \citenamefont {Gong}, \citenamefont {Quan}, \citenamefont
  {Jarzynski},\ and\ \citenamefont {Kim}}]{Smith2018}%
  \BibitemOpen
  \bibfield  {author} {\bibinfo {author} {\bibfnamefont {A.}~\bibnamefont
  {Smith}}, \bibinfo {author} {\bibfnamefont {Y.}~\bibnamefont {Lu}}, \bibinfo
  {author} {\bibfnamefont {S.}~\bibnamefont {An}}, \bibinfo {author}
  {\bibfnamefont {X.}~\bibnamefont {Zhang}}, \bibinfo {author} {\bibfnamefont
  {J.-N.}\ \bibnamefont {Zhang}}, \bibinfo {author} {\bibfnamefont
  {Z.}~\bibnamefont {Gong}}, \bibinfo {author} {\bibfnamefont {H.~T.}\
  \bibnamefont {Quan}}, \bibinfo {author} {\bibfnamefont {C.}~\bibnamefont
  {Jarzynski}},\ and\ \bibinfo {author} {\bibfnamefont {K.}~\bibnamefont
  {Kim}},\ }\bibfield  {title} {\bibinfo {title} {Verification of the quantum
  nonequilibrium work relation in the presence of decoherence},\ }\href
  {https://iopscience.iop.org/article/10.1088/1367-2630/aa9cd6} {\bibfield
  {journal} {\bibinfo  {journal} {New J. Phys.}\ }\textbf {\bibinfo {volume}
  {20}},\ \bibinfo {pages} {013008} (\bibinfo {year} {2018})}\BibitemShut
  {NoStop}%
\bibitem [{\citenamefont {An}\ \emph {et~al.}(2015)\citenamefont {An},
  \citenamefont {Zhang}, \citenamefont {Um}, \citenamefont {Lv}, \citenamefont
  {Lu}, \citenamefont {Zhang}, \citenamefont {Yin}, \citenamefont {Quan},\ and\
  \citenamefont {Kim}}]{An15}%
  \BibitemOpen
  \bibfield  {author} {\bibinfo {author} {\bibfnamefont {S.}~\bibnamefont
  {An}}, \bibinfo {author} {\bibfnamefont {J.-N.}\ \bibnamefont {Zhang}},
  \bibinfo {author} {\bibfnamefont {M.}~\bibnamefont {Um}}, \bibinfo {author}
  {\bibfnamefont {D.}~\bibnamefont {Lv}}, \bibinfo {author} {\bibfnamefont
  {Y.}~\bibnamefont {Lu}}, \bibinfo {author} {\bibfnamefont {J.}~\bibnamefont
  {Zhang}}, \bibinfo {author} {\bibfnamefont {Z.-Q.}\ \bibnamefont {Yin}},
  \bibinfo {author} {\bibfnamefont {H.~T.}\ \bibnamefont {Quan}},\ and\
  \bibinfo {author} {\bibfnamefont {K.}~\bibnamefont {Kim}},\ }\bibfield
  {title} {\bibinfo {title} {Experimental test of the quantum {J}arzynski
  equality with a trapped-ion system},\ }\href
  {https://www.nature.com/articles/nphys3197} {\bibfield  {journal} {\bibinfo
  {journal} {Nat. Phys.}\ }\textbf {\bibinfo {volume} {11}},\ \bibinfo {pages}
  {193} (\bibinfo {year} {2015})}\BibitemShut {NoStop}%
\bibitem [{\citenamefont {Campisi}\ \emph {et~al.}(2011)\citenamefont
  {Campisi}, \citenamefont {H\"{a}nggi},\ and\ \citenamefont
  {Talkner}}]{Campisi11}%
  \BibitemOpen
  \bibfield  {author} {\bibinfo {author} {\bibfnamefont {M.}~\bibnamefont
  {Campisi}}, \bibinfo {author} {\bibfnamefont {P.}~\bibnamefont
  {H\"{a}nggi}},\ and\ \bibinfo {author} {\bibfnamefont {P.}~\bibnamefont
  {Talkner}},\ }\bibfield  {title} {\bibinfo {title} {Colloquium: Quantum
  fluctuation relations: Foundations and applications},\ }\href
  {https://journals.aps.org/rmp/abstract/10.1103/RevModPhys.83.771} {\bibfield
  {journal} {\bibinfo  {journal} {Rev. Mod. Phys.}\ }\textbf {\bibinfo {volume}
  {83}},\ \bibinfo {pages} {771} (\bibinfo {year} {2011})}\BibitemShut
  {NoStop}%
\bibitem [{\citenamefont {Aguilar}\ \emph {et~al.}(2021)\citenamefont
  {Aguilar}, \citenamefont {Silva}, \citenamefont {Guimar{\~a}es},
  \citenamefont {Piera}, \citenamefont {C{\'e}leri},\ and\ \citenamefont
  {Landi}}]{aguilar2021two}%
  \BibitemOpen
  \bibfield  {author} {\bibinfo {author} {\bibfnamefont {G.~H.}\ \bibnamefont
  {Aguilar}}, \bibinfo {author} {\bibfnamefont {T.~L.}\ \bibnamefont {Silva}},
  \bibinfo {author} {\bibfnamefont {T.~E.}\ \bibnamefont {Guimar{\~a}es}},
  \bibinfo {author} {\bibfnamefont {R.~S.}\ \bibnamefont {Piera}}, \bibinfo
  {author} {\bibfnamefont {L.~C.}\ \bibnamefont {C{\'e}leri}},\ and\ \bibinfo
  {author} {\bibfnamefont {G.~T.}\ \bibnamefont {Landi}},\ }\bibfield  {title}
  {\bibinfo {title} {Two-point measurement of entropy production from the
  outcomes of a single experiment with correlated photon pairs},\ }\href
  {https://arxiv.org/abs/2108.03289} {\bibfield  {journal} {\bibinfo  {journal}
  {arXiv:2108.03289}\ } (\bibinfo {year} {2021})}\BibitemShut {NoStop}%
\bibitem [{\citenamefont {Hern{\'a}ndez-G{\'o}mez}\ \emph
  {et~al.}(2020)\citenamefont {Hern{\'a}ndez-G{\'o}mez}, \citenamefont
  {Gherardini}, \citenamefont {Poggiali}, \citenamefont {Cataliotti},
  \citenamefont {Trombettoni}, \citenamefont {Cappellaro},\ and\ \citenamefont
  {Fabbri}}]{hernandez2020experimental}%
  \BibitemOpen
  \bibfield  {author} {\bibinfo {author} {\bibfnamefont {S.}~\bibnamefont
  {Hern{\'a}ndez-G{\'o}mez}}, \bibinfo {author} {\bibfnamefont
  {S.}~\bibnamefont {Gherardini}}, \bibinfo {author} {\bibfnamefont
  {F.}~\bibnamefont {Poggiali}}, \bibinfo {author} {\bibfnamefont {F.~S.}\
  \bibnamefont {Cataliotti}}, \bibinfo {author} {\bibfnamefont
  {A.}~\bibnamefont {Trombettoni}}, \bibinfo {author} {\bibfnamefont
  {P.}~\bibnamefont {Cappellaro}},\ and\ \bibinfo {author} {\bibfnamefont
  {N.}~\bibnamefont {Fabbri}},\ }\bibfield  {title} {\bibinfo {title}
  {Experimental test of exchange fluctuation relations in an open quantum
  system},\ }\href {https://doi.org/10.1103/PhysRevResearch.2.023327}
  {\bibfield  {journal} {\bibinfo  {journal} {Phys. Rev. Research}\ }\textbf
  {\bibinfo {volume} {2}},\ \bibinfo {pages} {023327} (\bibinfo {year}
  {2020})}\BibitemShut {NoStop}%
\bibitem [{\citenamefont {Hern{\'a}ndez-G{\'o}mez}\ \emph
  {et~al.}(2021)\citenamefont {Hern{\'a}ndez-G{\'o}mez}, \citenamefont
  {Staudenmaier}, \citenamefont {Campisi},\ and\ \citenamefont
  {Fabbri}}]{hernandez2021experimental}%
  \BibitemOpen
  \bibfield  {author} {\bibinfo {author} {\bibfnamefont {S.}~\bibnamefont
  {Hern{\'a}ndez-G{\'o}mez}}, \bibinfo {author} {\bibfnamefont
  {N.}~\bibnamefont {Staudenmaier}}, \bibinfo {author} {\bibfnamefont
  {M.}~\bibnamefont {Campisi}},\ and\ \bibinfo {author} {\bibfnamefont
  {N.}~\bibnamefont {Fabbri}},\ }\bibfield  {title} {\bibinfo {title}
  {Experimental test of fluctuation relations for driven open quantum systems
  with an nv center},\ }\href {https://doi.org/10.1088/1367-2630/abfc6a}
  {\bibfield  {journal} {\bibinfo  {journal} {New J. Phys.}\ }\textbf {\bibinfo
  {volume} {23}},\ \bibinfo {pages} {065004} (\bibinfo {year}
  {2021})}\BibitemShut {NoStop}%
\bibitem [{\citenamefont {Albash}\ \emph {et~al.}(2013)\citenamefont {Albash},
  \citenamefont {Lidar}, \citenamefont {Marvian},\ and\ \citenamefont
  {Zanardi}}]{albash2013fluctuation}%
  \BibitemOpen
  \bibfield  {author} {\bibinfo {author} {\bibfnamefont {T.}~\bibnamefont
  {Albash}}, \bibinfo {author} {\bibfnamefont {D.~A.}\ \bibnamefont {Lidar}},
  \bibinfo {author} {\bibfnamefont {M.}~\bibnamefont {Marvian}},\ and\ \bibinfo
  {author} {\bibfnamefont {P.}~\bibnamefont {Zanardi}},\ }\bibfield  {title}
  {\bibinfo {title} {Fluctuation theorems for quantum processes},\ }\href
  {https://doi.org/10.1103/PhysRevE.88.032146} {\bibfield  {journal} {\bibinfo
  {journal} {Phys. Rev. E}\ }\textbf {\bibinfo {volume} {88}},\ \bibinfo
  {pages} {032146} (\bibinfo {year} {2013})}\BibitemShut {NoStop}%
\bibitem [{\citenamefont {Rastegin}\ and\ \citenamefont
  {{\.Z}yczkowski}(2014)}]{rastegin2014jarzynski}%
  \BibitemOpen
  \bibfield  {author} {\bibinfo {author} {\bibfnamefont {A.~E.}\ \bibnamefont
  {Rastegin}}\ and\ \bibinfo {author} {\bibfnamefont {K.}~\bibnamefont
  {{\.Z}yczkowski}},\ }\bibfield  {title} {\bibinfo {title} {Jarzynski equality
  for quantum stochastic maps},\ }\href
  {https://doi.org/10.1103/PhysRevE.89.012127} {\bibfield  {journal} {\bibinfo
  {journal} {Phys. Rev. E}\ }\textbf {\bibinfo {volume} {89}},\ \bibinfo
  {pages} {012127} (\bibinfo {year} {2014})}\BibitemShut {NoStop}%
\bibitem [{\citenamefont {Rastegin}(2013)}]{rastegin2013non}%
  \BibitemOpen
  \bibfield  {author} {\bibinfo {author} {\bibfnamefont {A.~E.}\ \bibnamefont
  {Rastegin}},\ }\bibfield  {title} {\bibinfo {title} {Non-equilibrium
  equalities with unital quantum channels},\ }\href
  {https://doi.org/10.1088/1742-5468/2013/06/P06016} {\bibfield  {journal}
  {\bibinfo  {journal} {J. Stat. Mech. Theory Exp.}\ }\textbf {\bibinfo
  {volume} {2013}},\ \bibinfo {pages} {P06016} (\bibinfo {year}
  {2013})}\BibitemShut {NoStop}%
\bibitem [{\citenamefont {Jarzynski}\ \emph {et~al.}(2015)\citenamefont
  {Jarzynski}, \citenamefont {Quan},\ and\ \citenamefont
  {Rahav}}]{jarzynski2015quantum}%
  \BibitemOpen
  \bibfield  {author} {\bibinfo {author} {\bibfnamefont {C.}~\bibnamefont
  {Jarzynski}}, \bibinfo {author} {\bibfnamefont {H.}~\bibnamefont {Quan}},\
  and\ \bibinfo {author} {\bibfnamefont {S.}~\bibnamefont {Rahav}},\ }\bibfield
   {title} {\bibinfo {title} {Quantum-classical correspondence principle for
  work distributions},\ }\href {https://doi.org/10.1103/PhysRevX.5.031038}
  {\bibfield  {journal} {\bibinfo  {journal} {Phys. Rev. X}\ }\textbf {\bibinfo
  {volume} {5}},\ \bibinfo {pages} {031038} (\bibinfo {year}
  {2015})}\BibitemShut {NoStop}%
\bibitem [{\citenamefont {Perarnau-Llobet}\ \emph {et~al.}(2017)\citenamefont
  {Perarnau-Llobet}, \citenamefont {B{\"a}umer}, \citenamefont {Hovhannisyan},
  \citenamefont {Huber},\ and\ \citenamefont {Acin}}]{perarnau2017no}%
  \BibitemOpen
  \bibfield  {author} {\bibinfo {author} {\bibfnamefont {M.}~\bibnamefont
  {Perarnau-Llobet}}, \bibinfo {author} {\bibfnamefont {E.}~\bibnamefont
  {B{\"a}umer}}, \bibinfo {author} {\bibfnamefont {K.~V.}\ \bibnamefont
  {Hovhannisyan}}, \bibinfo {author} {\bibfnamefont {M.}~\bibnamefont
  {Huber}},\ and\ \bibinfo {author} {\bibfnamefont {A.}~\bibnamefont {Acin}},\
  }\bibfield  {title} {\bibinfo {title} {No-go theorem for the characterization
  of work fluctuations in coherent quantum systems},\ }\href
  {https://doi.org/10.1103/PhysRevLett.118.070601} {\bibfield  {journal}
  {\bibinfo  {journal} {Phys. Rev. Lett.}\ }\textbf {\bibinfo {volume} {118}},\
  \bibinfo {pages} {070601} (\bibinfo {year} {2017})}\BibitemShut {NoStop}%
\bibitem [{\citenamefont {Buscemi}\ and\ \citenamefont
  {Scarani}(2021)}]{buscemi2021fluctuation}%
  \BibitemOpen
  \bibfield  {author} {\bibinfo {author} {\bibfnamefont {F.}~\bibnamefont
  {Buscemi}}\ and\ \bibinfo {author} {\bibfnamefont {V.}~\bibnamefont
  {Scarani}},\ }\bibfield  {title} {\bibinfo {title} {Fluctuation theorems from
  bayesian retrodiction},\ }\href {https://doi.org/10.1103/PhysRevE.103.052111}
  {\bibfield  {journal} {\bibinfo  {journal} {Phys. Rev. E}\ }\textbf {\bibinfo
  {volume} {103}},\ \bibinfo {pages} {052111} (\bibinfo {year}
  {2021})}\BibitemShut {NoStop}%
\bibitem [{\citenamefont {Micadei}\ \emph {et~al.}(2020)\citenamefont
  {Micadei}, \citenamefont {Landi},\ and\ \citenamefont
  {Lutz}}]{micadei2020quantum}%
  \BibitemOpen
  \bibfield  {author} {\bibinfo {author} {\bibfnamefont {K.}~\bibnamefont
  {Micadei}}, \bibinfo {author} {\bibfnamefont {G.~T.}\ \bibnamefont {Landi}},\
  and\ \bibinfo {author} {\bibfnamefont {E.}~\bibnamefont {Lutz}},\ }\bibfield
  {title} {\bibinfo {title} {Quantum fluctuation theorems beyond two-point
  measurements},\ }\href {https://doi.org/10.1103/PhysRevLett.124.090602}
  {\bibfield  {journal} {\bibinfo  {journal} {Phys. Rev. Lett.}\ }\textbf
  {\bibinfo {volume} {124}},\ \bibinfo {pages} {090602} (\bibinfo {year}
  {2020})}\BibitemShut {NoStop}%
\bibitem [{\citenamefont {Levy}\ and\ \citenamefont
  {Lostaglio}(2020)}]{levy2020quasiprobability}%
  \BibitemOpen
  \bibfield  {author} {\bibinfo {author} {\bibfnamefont {A.}~\bibnamefont
  {Levy}}\ and\ \bibinfo {author} {\bibfnamefont {M.}~\bibnamefont
  {Lostaglio}},\ }\bibfield  {title} {\bibinfo {title} {Quasiprobability
  distribution for heat fluctuations in the quantum regime},\ }\href
  {https://doi.org/10.1103/PRXQuantum.1.010309} {\bibfield  {journal} {\bibinfo
   {journal} {PRX Quantum}\ }\textbf {\bibinfo {volume} {1}},\ \bibinfo {pages}
  {010309} (\bibinfo {year} {2020})}\BibitemShut {NoStop}%
\bibitem [{\citenamefont {Lostaglio}(2018)}]{lostaglio2018quantum}%
  \BibitemOpen
  \bibfield  {author} {\bibinfo {author} {\bibfnamefont {M.}~\bibnamefont
  {Lostaglio}},\ }\bibfield  {title} {\bibinfo {title} {Quantum fluctuation
  theorems, contextuality, and work quasiprobabilities},\ }\href
  {https://doi.org/10.1103/PhysRevLett.120.040602} {\bibfield  {journal}
  {\bibinfo  {journal} {Phys. Rev. Lett.}\ }\textbf {\bibinfo {volume} {120}},\
  \bibinfo {pages} {040602} (\bibinfo {year} {2018})}\BibitemShut {NoStop}%
\bibitem [{\citenamefont {Buscemi}\ \emph {et~al.}(2016)\citenamefont
  {Buscemi}, \citenamefont {Das},\ and\ \citenamefont
  {Wilde}}]{buscemi2016approximate}%
  \BibitemOpen
  \bibfield  {author} {\bibinfo {author} {\bibfnamefont {F.}~\bibnamefont
  {Buscemi}}, \bibinfo {author} {\bibfnamefont {S.}~\bibnamefont {Das}},\ and\
  \bibinfo {author} {\bibfnamefont {M.~M.}\ \bibnamefont {Wilde}},\ }\bibfield
  {title} {\bibinfo {title} {Approximate reversibility in the context of
  entropy gain, information gain, and complete positivity},\ }\href
  {https://doi.org/10.1103/PhysRevA.93.062314} {\bibfield  {journal} {\bibinfo
  {journal} {Phys. Rev. A}\ }\textbf {\bibinfo {volume} {93}},\ \bibinfo
  {pages} {062314} (\bibinfo {year} {2016})}\BibitemShut {NoStop}%
\bibitem [{\citenamefont {Das}\ \emph {et~al.}(2018)\citenamefont {Das},
  \citenamefont {Khatri}, \citenamefont {Siopsis},\ and\ \citenamefont
  {Wilde}}]{das2018fundamental}%
  \BibitemOpen
  \bibfield  {author} {\bibinfo {author} {\bibfnamefont {S.}~\bibnamefont
  {Das}}, \bibinfo {author} {\bibfnamefont {S.}~\bibnamefont {Khatri}},
  \bibinfo {author} {\bibfnamefont {G.}~\bibnamefont {Siopsis}},\ and\ \bibinfo
  {author} {\bibfnamefont {M.~M.}\ \bibnamefont {Wilde}},\ }\bibfield  {title}
  {\bibinfo {title} {Fundamental limits on quantum dynamics based on entropy
  change},\ }\href {https://doi.org/https://doi.org/10.1063/1.4997044}
  {\bibfield  {journal} {\bibinfo  {journal} {Journal of Mathematical Physics}\
  }\textbf {\bibinfo {volume} {59}},\ \bibinfo {pages} {012205} (\bibinfo
  {year} {2018})}\BibitemShut {NoStop}%
\bibitem [{\citenamefont {Shangnan}\ and\ \citenamefont
  {Wang}(2021)}]{zhou_and_wang_2021quantum}%
  \BibitemOpen
  \bibfield  {author} {\bibinfo {author} {\bibfnamefont {Z.}~\bibnamefont
  {Shangnan}}\ and\ \bibinfo {author} {\bibfnamefont {Y.}~\bibnamefont
  {Wang}},\ }\bibfield  {title} {\bibinfo {title} {Quantum {C}ross {E}ntropy
  and {M}aximum {L}ikelihood {P}rinciple},\ }\href
  {https://arxiv.org/pdf/2102.11887.pdf} {\bibfield  {journal} {\bibinfo
  {journal} {arXiv:2102.11887}\ } (\bibinfo {year} {2021})}\BibitemShut
  {NoStop}%
\bibitem [{\citenamefont {Shangnan}(2021)}]{zhou2021quantum}%
  \BibitemOpen
  \bibfield  {author} {\bibinfo {author} {\bibfnamefont {Z.}~\bibnamefont
  {Shangnan}},\ }\bibfield  {title} {\bibinfo {title} {Quantum data compression
  and quantum cross entropy},\ }\href {https://arxiv.org/abs/2106.13823}
  {\bibfield  {journal} {\bibinfo  {journal} {arXiv:2106.13823}\ } (\bibinfo
  {year} {2021})}\BibitemShut {NoStop}%
\bibitem [{\citenamefont {Romero}\ \emph {et~al.}(2017)\citenamefont {Romero},
  \citenamefont {Olson},\ and\ \citenamefont {Aspuru-Guzik}}]{Romero17}%
  \BibitemOpen
  \bibfield  {author} {\bibinfo {author} {\bibfnamefont {J.}~\bibnamefont
  {Romero}}, \bibinfo {author} {\bibfnamefont {J.~P.}\ \bibnamefont {Olson}},\
  and\ \bibinfo {author} {\bibfnamefont {A.}~\bibnamefont {Aspuru-Guzik}},\
  }\bibfield  {title} {\bibinfo {title} {{Quantum autoencoders for efficient
  compression of quantum data}},\ }\href
  {https://doi.org/10.1088/2058-9565/aa8072} {\bibfield  {journal} {\bibinfo
  {journal} {Quantum Sci. Technol.}\ }\textbf {\bibinfo {volume} {2}},\
  \bibinfo {pages} {045001} (\bibinfo {year} {2017})}\BibitemShut {NoStop}%
\bibitem [{\citenamefont {Wan}\ \emph {et~al.}(2017)\citenamefont {Wan},
  \citenamefont {Dahlsten}, \citenamefont {Kristj{\'a}nsson}, \citenamefont
  {Gardner},\ and\ \citenamefont {Kim}}]{wan2017quantum}%
  \BibitemOpen
  \bibfield  {author} {\bibinfo {author} {\bibfnamefont {K.~H.}\ \bibnamefont
  {Wan}}, \bibinfo {author} {\bibfnamefont {O.}~\bibnamefont {Dahlsten}},
  \bibinfo {author} {\bibfnamefont {H.}~\bibnamefont {Kristj{\'a}nsson}},
  \bibinfo {author} {\bibfnamefont {R.}~\bibnamefont {Gardner}},\ and\ \bibinfo
  {author} {\bibfnamefont {M.~S.}\ \bibnamefont {Kim}},\ }\bibfield  {title}
  {\bibinfo {title} {Quantum generalisation of feedforward neural networks},\
  }\href {https://doi.org/10.1038/s41534-017-0032-4} {\bibfield  {journal}
  {\bibinfo  {journal} {npj Quantum Inf.}\ }\textbf {\bibinfo {volume} {3}},\
  \bibinfo {pages} {1} (\bibinfo {year} {2017})}\BibitemShut {NoStop}%
\bibitem [{\citenamefont {Preskill}(2018)}]{Preskill18}%
  \BibitemOpen
  \bibfield  {author} {\bibinfo {author} {\bibfnamefont {J.}~\bibnamefont
  {Preskill}},\ }\bibfield  {title} {\bibinfo {title} {Quantum computing in the
  {NISQ} era and beyond},\ }\href {https://doi.org/10.22331/q-2018-08-06-79}
  {\bibfield  {journal} {\bibinfo  {journal} {Quantum}\ }\textbf {\bibinfo
  {volume} {2}},\ \bibinfo {pages} {79} (\bibinfo {year} {2018})}\BibitemShut
  {NoStop}%
\bibitem [{\citenamefont {McClean}\ \emph {et~al.}(2016)\citenamefont
  {McClean}, \citenamefont {Romero}, \citenamefont {Babbush},\ and\
  \citenamefont {Aspuru-Guzik}}]{mcclean2016theory}%
  \BibitemOpen
  \bibfield  {author} {\bibinfo {author} {\bibfnamefont {J.~R.}\ \bibnamefont
  {McClean}}, \bibinfo {author} {\bibfnamefont {J.}~\bibnamefont {Romero}},
  \bibinfo {author} {\bibfnamefont {R.}~\bibnamefont {Babbush}},\ and\ \bibinfo
  {author} {\bibfnamefont {A.}~\bibnamefont {Aspuru-Guzik}},\ }\bibfield
  {title} {\bibinfo {title} {The theory of variational hybrid quantum-classical
  algorithms},\ }\href {https://doi.org/10.1007/978-94-015-8330-5_4} {\bibfield
   {journal} {\bibinfo  {journal} {New J. Phys.}\ }\textbf {\bibinfo {volume}
  {18}},\ \bibinfo {pages} {023023} (\bibinfo {year} {2016})}\BibitemShut
  {NoStop}%
\bibitem [{\citenamefont {Jones}\ \emph {et~al.}(2019)\citenamefont {Jones},
  \citenamefont {Endo}, \citenamefont {McArdle}, \citenamefont {Yuan},\ and\
  \citenamefont {Benjamin}}]{jones2019variational}%
  \BibitemOpen
  \bibfield  {author} {\bibinfo {author} {\bibfnamefont {T.}~\bibnamefont
  {Jones}}, \bibinfo {author} {\bibfnamefont {S.}~\bibnamefont {Endo}},
  \bibinfo {author} {\bibfnamefont {S.}~\bibnamefont {McArdle}}, \bibinfo
  {author} {\bibfnamefont {X.}~\bibnamefont {Yuan}},\ and\ \bibinfo {author}
  {\bibfnamefont {S.~C.}\ \bibnamefont {Benjamin}},\ }\bibfield  {title}
  {\bibinfo {title} {Variational quantum algorithms for discovering hamiltonian
  spectra},\ }\href
  {https://journals.aps.org/pra/abstract/10.1103/PhysRevA.99.062304} {\bibfield
   {journal} {\bibinfo  {journal} {Phys. Rev. A}\ }\textbf {\bibinfo {volume}
  {99}},\ \bibinfo {pages} {062304} (\bibinfo {year} {2019})}\BibitemShut
  {NoStop}%
\bibitem [{\citenamefont {Nakanishi}\ \emph {et~al.}(2020)\citenamefont
  {Nakanishi}, \citenamefont {Fujii},\ and\ \citenamefont
  {Todo}}]{nakanishi2020sequential}%
  \BibitemOpen
  \bibfield  {author} {\bibinfo {author} {\bibfnamefont {K.~M.}\ \bibnamefont
  {Nakanishi}}, \bibinfo {author} {\bibfnamefont {K.}~\bibnamefont {Fujii}},\
  and\ \bibinfo {author} {\bibfnamefont {S.}~\bibnamefont {Todo}},\ }\bibfield
  {title} {\bibinfo {title} {Sequential minimal optimization for
  quantum-classical hybrid algorithms},\ }\href
  {https://doi.org/10.1103/PhysRevResearch.2.043158} {\bibfield  {journal}
  {\bibinfo  {journal} {Phys. Rev. Research}\ }\textbf {\bibinfo {volume}
  {2}},\ \bibinfo {pages} {043158} (\bibinfo {year} {2020})}\BibitemShut
  {NoStop}%
\bibitem [{\citenamefont {Cerezo}\ \emph
  {et~al.}(2021{\natexlab{a}})\citenamefont {Cerezo}, \citenamefont
  {Arrasmith}, \citenamefont {Babbush}, \citenamefont {Benjamin}, \citenamefont
  {Endo}, \citenamefont {Fujii}, \citenamefont {McClean}, \citenamefont
  {Mitarai}, \citenamefont {Yuan}, \citenamefont {Cincio},\ and\ \citenamefont
  {Coles}}]{cerezo2021variational}%
  \BibitemOpen
  \bibfield  {author} {\bibinfo {author} {\bibfnamefont {M.}~\bibnamefont
  {Cerezo}}, \bibinfo {author} {\bibfnamefont {A.}~\bibnamefont {Arrasmith}},
  \bibinfo {author} {\bibfnamefont {R.}~\bibnamefont {Babbush}}, \bibinfo
  {author} {\bibfnamefont {S.~C.}\ \bibnamefont {Benjamin}}, \bibinfo {author}
  {\bibfnamefont {S.}~\bibnamefont {Endo}}, \bibinfo {author} {\bibfnamefont
  {K.}~\bibnamefont {Fujii}}, \bibinfo {author} {\bibfnamefont {J.~R.}\
  \bibnamefont {McClean}}, \bibinfo {author} {\bibfnamefont {K.}~\bibnamefont
  {Mitarai}}, \bibinfo {author} {\bibfnamefont {X.}~\bibnamefont {Yuan}},
  \bibinfo {author} {\bibfnamefont {L.}~\bibnamefont {Cincio}},\ and\ \bibinfo
  {author} {\bibfnamefont {P.~J.}\ \bibnamefont {Coles}},\ }\bibfield  {title}
  {\bibinfo {title} {Variational quantum algorithms},\ }\href
  {https://doi.org/10.1038/s42254-021-00348-9} {\bibfield  {journal} {\bibinfo
  {journal} {Nat. Rev. Phys.}\ }\textbf {\bibinfo {volume} {3}},\ \bibinfo
  {pages} {625} (\bibinfo {year} {2021}{\natexlab{a}})}\BibitemShut {NoStop}%
\bibitem [{\citenamefont {Bharti}\ \emph {et~al.}(2022)\citenamefont {Bharti},
  \citenamefont {Cervera-Lierta}, \citenamefont {Kyaw}, \citenamefont {Haug},
  \citenamefont {Alperin-Lea}, \citenamefont {Anand}, \citenamefont {Degroote},
  \citenamefont {Heimonen}, \citenamefont {Kottmann}, \citenamefont {Menke}
  \emph {et~al.}}]{bharti2022noisy}%
  \BibitemOpen
  \bibfield  {author} {\bibinfo {author} {\bibfnamefont {K.}~\bibnamefont
  {Bharti}}, \bibinfo {author} {\bibfnamefont {A.}~\bibnamefont
  {Cervera-Lierta}}, \bibinfo {author} {\bibfnamefont {T.~H.}\ \bibnamefont
  {Kyaw}}, \bibinfo {author} {\bibfnamefont {T.}~\bibnamefont {Haug}}, \bibinfo
  {author} {\bibfnamefont {S.}~\bibnamefont {Alperin-Lea}}, \bibinfo {author}
  {\bibfnamefont {A.}~\bibnamefont {Anand}}, \bibinfo {author} {\bibfnamefont
  {M.}~\bibnamefont {Degroote}}, \bibinfo {author} {\bibfnamefont
  {H.}~\bibnamefont {Heimonen}}, \bibinfo {author} {\bibfnamefont {J.~S.}\
  \bibnamefont {Kottmann}}, \bibinfo {author} {\bibfnamefont {T.}~\bibnamefont
  {Menke}}, \emph {et~al.},\ }\bibfield  {title} {\bibinfo {title} {Noisy
  intermediate-scale quantum algorithms},\ }\href
  {https://doi.org/10.1103/RevModPhys.94.015004} {\bibfield  {journal}
  {\bibinfo  {journal} {Rev. Mod. Phys.}\ }\textbf {\bibinfo {volume} {94}},\
  \bibinfo {pages} {015004} (\bibinfo {year} {2022})}\BibitemShut {NoStop}%
\bibitem [{\citenamefont {Deffner}(2017)}]{deffner2017kibble}%
  \BibitemOpen
  \bibfield  {author} {\bibinfo {author} {\bibfnamefont {S.}~\bibnamefont
  {Deffner}},\ }\bibfield  {title} {\bibinfo {title} {Kibble-{Z}urek scaling of
  the irreversible entropy production},\ }\href
  {https://doi.org/10.1103/PhysRevE.96.052125} {\bibfield  {journal} {\bibinfo
  {journal} {Phys. Rev. E}\ }\textbf {\bibinfo {volume} {96}},\ \bibinfo
  {pages} {052125} (\bibinfo {year} {2017})}\BibitemShut {NoStop}%
\bibitem [{\citenamefont {Wilde}(2013)}]{wilde2013quantum}%
  \BibitemOpen
  \bibfield  {author} {\bibinfo {author} {\bibfnamefont {M.~M.}\ \bibnamefont
  {Wilde}},\ }\href@noop {} {\emph {\bibinfo {title} {Quantum information
  theory}}}\ (\bibinfo  {publisher} {Cambridge University Press},\ \bibinfo
  {year} {2013})\BibitemShut {NoStop}%
\bibitem [{Note1()}]{Note1}%
  \BibitemOpen
  \bibinfo {note} {{The OTM scheme has been utilized to explore work and heat
  in the open quantum system~\cite {Sone20a}, its classical
  correspondence~\cite {Sone21b}, heat exchange~\cite {sone2022heat}, and work
  as an external observable~\cite {Beyer2020}. Particularly, a second-law-like
  inequality involving the guessed heat introduced in Ref.~\cite {Sone20a} can
  be derived by using $\setbox \z@ \hbox {\mathsurround \z@ $\textstyle
  P$}\mathaccent "0365{P}(\sigma )$ (See Appendix.~\ref
  {app:2ndLawGuessedHeat}).}}\BibitemShut {Stop}%
\bibitem [{Note2()}]{Note2}%
  \BibitemOpen
  \bibinfo {note} {{Note that our main claims are the derivation of the general
  integrated fluctuation theorems, which hold for any states and quantum
  channels, and its potential of characterizing the quantum protocol with the
  quantum cross entropy. We leave the comparison between the tight bound
  derived in Refs.~\cite {das2018fundamental,buscemi2016approximate} and our
  bound derived in the OTM scheme as an open problem.}}\BibitemShut {Stop}%
\bibitem [{\citenamefont {Nielsen}\ and\ \citenamefont
  {Chuang}(2010)}]{Nielsen}%
  \BibitemOpen
  \bibfield  {author} {\bibinfo {author} {\bibfnamefont {M.~A.}\ \bibnamefont
  {Nielsen}}\ and\ \bibinfo {author} {\bibfnamefont {I.~L.}\ \bibnamefont
  {Chuang}},\ }\href@noop {} {\emph {\bibinfo {title} {Quantum Computation and
  Quantum Information: 10th Anniversary Edition}}},\ \bibinfo {edition} {10th}\
  ed.\ (\bibinfo  {publisher} {Cambridge University Press},\ \bibinfo {address}
  {New York, NY, USA},\ \bibinfo {year} {2010})\BibitemShut {NoStop}%
\bibitem [{\citenamefont {Kingma}\ and\ \citenamefont
  {Welling}(2019)}]{Kingma_Book_Autoencoder}%
  \BibitemOpen
  \bibfield  {author} {\bibinfo {author} {\bibfnamefont {D.~P.}\ \bibnamefont
  {Kingma}}\ and\ \bibinfo {author} {\bibfnamefont {M.}~\bibnamefont
  {Welling}},\ }\bibfield  {title} {\bibinfo {title} {{A}n {I}ntroduction to
  {V}ariational {A}utoencoders},\ }\href {https://doi.org/10.1561/2200000056}
  {\bibfield  {journal} {\bibinfo  {journal} {Found. Trends Mach. Learn.}\
  }\textbf {\bibinfo {volume} {12}},\ \bibinfo {pages} {307} (\bibinfo {year}
  {2019})}\BibitemShut {NoStop}%
\bibitem [{\citenamefont {Bravo-Prieto}(2021)}]{Bravo_Prieto_2021}%
  \BibitemOpen
  \bibfield  {author} {\bibinfo {author} {\bibfnamefont {C.}~\bibnamefont
  {Bravo-Prieto}},\ }\bibfield  {title} {\bibinfo {title} {Quantum autoencoders
  with enhanced data encoding},\ }\href
  {https://doi.org/10.1088/2632-2153/ac0616} {\bibfield  {journal} {\bibinfo
  {journal} {Mach. Learn.: Sci. Technol.}\ }\textbf {\bibinfo {volume} {2}},\
  \bibinfo {pages} {035028} (\bibinfo {year} {2021})}\BibitemShut {NoStop}%
\bibitem [{\citenamefont {Bondarenko}\ and\ \citenamefont
  {Feldmann}(2020)}]{Bondarenko2019}%
  \BibitemOpen
  \bibfield  {author} {\bibinfo {author} {\bibfnamefont {D.}~\bibnamefont
  {Bondarenko}}\ and\ \bibinfo {author} {\bibfnamefont {P.}~\bibnamefont
  {Feldmann}},\ }\bibfield  {title} {\bibinfo {title} {Quantum autoencoders to
  denoise quantum data},\ }\href
  {https://doi.org/10.1103/PhysRevLett.124.130502} {\bibfield  {journal}
  {\bibinfo  {journal} {Phys. Rev. Lett.}\ }\textbf {\bibinfo {volume} {124}},\
  \bibinfo {pages} {130502} (\bibinfo {year} {2020})}\BibitemShut {NoStop}%
\bibitem [{\citenamefont {Locher}\ \emph {et~al.}(2022)\citenamefont {Locher},
  \citenamefont {Cardarelli},\ and\ \citenamefont
  {M{\"u}ller}}]{locher2022quantum}%
  \BibitemOpen
  \bibfield  {author} {\bibinfo {author} {\bibfnamefont {D.~F.}\ \bibnamefont
  {Locher}}, \bibinfo {author} {\bibfnamefont {L.}~\bibnamefont {Cardarelli}},\
  and\ \bibinfo {author} {\bibfnamefont {M.}~\bibnamefont {M{\"u}ller}},\
  }\bibfield  {title} {\bibinfo {title} {Quantum error correction with quantum
  autoencoders},\ }\href {https://arxiv.org/abs/2202.00555v1} {\bibfield
  {journal} {\bibinfo  {journal} {arXiv:2202.00555}\ } (\bibinfo {year}
  {2022})}\BibitemShut {NoStop}%
\bibitem [{\citenamefont {Cao}\ and\ \citenamefont {Wang}(2021)}]{Cao_2021}%
  \BibitemOpen
  \bibfield  {author} {\bibinfo {author} {\bibfnamefont {C.}~\bibnamefont
  {Cao}}\ and\ \bibinfo {author} {\bibfnamefont {X.}~\bibnamefont {Wang}},\
  }\bibfield  {title} {\bibinfo {title} {Noise-assisted quantum autoencoder},\
  }\href {https://doi.org/10.1103/physrevapplied.15.054012} {\bibfield
  {journal} {\bibinfo  {journal} {Phys. Rev. Appl.}\ }\textbf {\bibinfo
  {volume} {15}},\ \bibinfo {pages} {054012} (\bibinfo {year}
  {2021})}\BibitemShut {NoStop}%
\bibitem [{\citenamefont {Steinbrecher}\ \emph {et~al.}(2019)\citenamefont
  {Steinbrecher}, \citenamefont {Olson}, \citenamefont {Englund},\ and\
  \citenamefont {Carolan}}]{steinbrecher2019quantum}%
  \BibitemOpen
  \bibfield  {author} {\bibinfo {author} {\bibfnamefont {G.~R.}\ \bibnamefont
  {Steinbrecher}}, \bibinfo {author} {\bibfnamefont {J.~P.}\ \bibnamefont
  {Olson}}, \bibinfo {author} {\bibfnamefont {D.}~\bibnamefont {Englund}},\
  and\ \bibinfo {author} {\bibfnamefont {J.}~\bibnamefont {Carolan}},\
  }\bibfield  {title} {\bibinfo {title} {Quantum optical neural networks},\
  }\href {https://doi.org/10.1038/s41534-019-0174-7} {\bibfield  {journal}
  {\bibinfo  {journal} {npj Quantum Inf.}\ }\textbf {\bibinfo {volume} {5}},\
  \bibinfo {pages} {1} (\bibinfo {year} {2019})}\BibitemShut {NoStop}%
\bibitem [{\citenamefont {Du}\ and\ \citenamefont
  {Tao}(2021)}]{du2021exploring}%
  \BibitemOpen
  \bibfield  {author} {\bibinfo {author} {\bibfnamefont {Y.}~\bibnamefont
  {Du}}\ and\ \bibinfo {author} {\bibfnamefont {D.}~\bibnamefont {Tao}},\
  }\bibfield  {title} {\bibinfo {title} {On exploring practical potentials of
  quantum auto-encoder with advantages},\ }\href
  {https://arxiv.org/abs/2106.15432} {\bibfield  {journal} {\bibinfo  {journal}
  {arXiv:2106.15432}\ } (\bibinfo {year} {2021})}\BibitemShut {NoStop}%
\bibitem [{\citenamefont {Pepper}\ \emph {et~al.}(2019)\citenamefont {Pepper},
  \citenamefont {Tischler},\ and\ \citenamefont
  {Pryde}}]{pepper2019experimental}%
  \BibitemOpen
  \bibfield  {author} {\bibinfo {author} {\bibfnamefont {A.}~\bibnamefont
  {Pepper}}, \bibinfo {author} {\bibfnamefont {N.}~\bibnamefont {Tischler}},\
  and\ \bibinfo {author} {\bibfnamefont {G.~J.}\ \bibnamefont {Pryde}},\
  }\bibfield  {title} {\bibinfo {title} {Experimental realization of a quantum
  autoencoder: The compression of qutrits via machine learning},\ }\href
  {https://doi.org/10.1103/PhysRevLett.122.060501} {\bibfield  {journal}
  {\bibinfo  {journal} {Phys. Rev. Lett.}\ }\textbf {\bibinfo {volume} {122}},\
  \bibinfo {pages} {060501} (\bibinfo {year} {2019})}\BibitemShut {NoStop}%
\bibitem [{\citenamefont {Mangini}\ \emph {et~al.}(2022)\citenamefont
  {Mangini}, \citenamefont {Marruzzo}, \citenamefont {Piantanida},
  \citenamefont {Gerace}, \citenamefont {Bajoni},\ and\ \citenamefont
  {Macchiavello}}]{mangini2022quantum}%
  \BibitemOpen
  \bibfield  {author} {\bibinfo {author} {\bibfnamefont {S.}~\bibnamefont
  {Mangini}}, \bibinfo {author} {\bibfnamefont {A.}~\bibnamefont {Marruzzo}},
  \bibinfo {author} {\bibfnamefont {M.}~\bibnamefont {Piantanida}}, \bibinfo
  {author} {\bibfnamefont {D.}~\bibnamefont {Gerace}}, \bibinfo {author}
  {\bibfnamefont {D.}~\bibnamefont {Bajoni}},\ and\ \bibinfo {author}
  {\bibfnamefont {C.}~\bibnamefont {Macchiavello}},\ }\bibfield  {title}
  {\bibinfo {title} {Quantum neural network autoencoder and classifier applied
  to an industrial case study},\ }\href {https://arxiv.org/abs/2205.04127}
  {\bibfield  {journal} {\bibinfo  {journal} {arXiv:2205.04127}\ } (\bibinfo
  {year} {2022})}\BibitemShut {NoStop}%
\bibitem [{\citenamefont {Patel}\ \emph {et~al.}(2022)\citenamefont {Patel},
  \citenamefont {Collis}, \citenamefont {Duong}, \citenamefont {Koch},
  \citenamefont {Cutugno}, \citenamefont {Wessing},\ and\ \citenamefont
  {Alsing}}]{patel2022information}%
  \BibitemOpen
  \bibfield  {author} {\bibinfo {author} {\bibfnamefont {S.}~\bibnamefont
  {Patel}}, \bibinfo {author} {\bibfnamefont {B.}~\bibnamefont {Collis}},
  \bibinfo {author} {\bibfnamefont {W.}~\bibnamefont {Duong}}, \bibinfo
  {author} {\bibfnamefont {D.}~\bibnamefont {Koch}}, \bibinfo {author}
  {\bibfnamefont {M.}~\bibnamefont {Cutugno}}, \bibinfo {author} {\bibfnamefont
  {L.}~\bibnamefont {Wessing}},\ and\ \bibinfo {author} {\bibfnamefont
  {P.}~\bibnamefont {Alsing}},\ }\bibfield  {title} {\bibinfo {title}
  {Information loss and run time from practical application of quantum data
  compression},\ }\href {https://arxiv.org/abs/2203.11332} {\bibfield
  {journal} {\bibinfo  {journal} {arXiv preprint arXiv:2203.11332}\ } (\bibinfo
  {year} {2022})}\BibitemShut {NoStop}%
\bibitem [{\citenamefont {Ngairangbam}\ \emph {et~al.}(2022)\citenamefont
  {Ngairangbam}, \citenamefont {Spannowsky},\ and\ \citenamefont
  {Takeuchi}}]{ngairangbam2022anomaly}%
  \BibitemOpen
  \bibfield  {author} {\bibinfo {author} {\bibfnamefont {V.~S.}\ \bibnamefont
  {Ngairangbam}}, \bibinfo {author} {\bibfnamefont {M.}~\bibnamefont
  {Spannowsky}},\ and\ \bibinfo {author} {\bibfnamefont {M.}~\bibnamefont
  {Takeuchi}},\ }\bibfield  {title} {\bibinfo {title} {Anomaly detection in
  high-energy physics using a quantum autoencoder},\ }\href
  {https://doi.org/10.1103/PhysRevD.105.095004} {\bibfield  {journal} {\bibinfo
   {journal} {Phys. Rev. D}\ }\textbf {\bibinfo {volume} {105}},\ \bibinfo
  {pages} {095004} (\bibinfo {year} {2022})}\BibitemShut {NoStop}%
\bibitem [{\citenamefont {Huang}\ \emph {et~al.}(2020)\citenamefont {Huang},
  \citenamefont {Ma}, \citenamefont {Yin}, \citenamefont {Tang}, \citenamefont
  {Dong}, \citenamefont {Chen}, \citenamefont {Xiang}, \citenamefont {Li},\
  and\ \citenamefont {Guo}}]{huang2020realization}%
  \BibitemOpen
  \bibfield  {author} {\bibinfo {author} {\bibfnamefont {C.-J.}\ \bibnamefont
  {Huang}}, \bibinfo {author} {\bibfnamefont {H.}~\bibnamefont {Ma}}, \bibinfo
  {author} {\bibfnamefont {Q.}~\bibnamefont {Yin}}, \bibinfo {author}
  {\bibfnamefont {J.-F.}\ \bibnamefont {Tang}}, \bibinfo {author}
  {\bibfnamefont {D.}~\bibnamefont {Dong}}, \bibinfo {author} {\bibfnamefont
  {C.}~\bibnamefont {Chen}}, \bibinfo {author} {\bibfnamefont {G.-Y.}\
  \bibnamefont {Xiang}}, \bibinfo {author} {\bibfnamefont {C.-F.}\ \bibnamefont
  {Li}},\ and\ \bibinfo {author} {\bibfnamefont {G.-C.}\ \bibnamefont {Guo}},\
  }\bibfield  {title} {\bibinfo {title} {Realization of a quantum autoencoder
  for lossless compression of quantum data},\ }\href
  {https://doi.org/10.1103/PhysRevA.102.032412} {\bibfield  {journal} {\bibinfo
   {journal} {Phys. Rev. A}\ }\textbf {\bibinfo {volume} {102}},\ \bibinfo
  {pages} {032412} (\bibinfo {year} {2020})}\BibitemShut {NoStop}%
\bibitem [{\citenamefont {Ma}\ \emph {et~al.}(2020)\citenamefont {Ma},
  \citenamefont {Huang}, \citenamefont {Chen}, \citenamefont {Dong},
  \citenamefont {Wang}, \citenamefont {Wu},\ and\ \citenamefont
  {Xiang}}]{ma2020compression}%
  \BibitemOpen
  \bibfield  {author} {\bibinfo {author} {\bibfnamefont {H.}~\bibnamefont
  {Ma}}, \bibinfo {author} {\bibfnamefont {C.-J.}\ \bibnamefont {Huang}},
  \bibinfo {author} {\bibfnamefont {C.}~\bibnamefont {Chen}}, \bibinfo {author}
  {\bibfnamefont {D.}~\bibnamefont {Dong}}, \bibinfo {author} {\bibfnamefont
  {Y.}~\bibnamefont {Wang}}, \bibinfo {author} {\bibfnamefont {R.-B.}\
  \bibnamefont {Wu}},\ and\ \bibinfo {author} {\bibfnamefont {G.-Y.}\
  \bibnamefont {Xiang}},\ }\bibfield  {title} {\bibinfo {title} {On compression
  rate of quantum autoencoders: Control design, numerical and experimental
  realization},\ }\href {https://arxiv.org/abs/2005.11149} {\bibfield
  {journal} {\bibinfo  {journal} {arXiv:2005.11149}\ } (\bibinfo {year}
  {2020})}\BibitemShut {NoStop}%
\bibitem [{\citenamefont {Buscemi}\ and\ \citenamefont
  {Horodecki}(2009)}]{buscemi2009towards}%
  \BibitemOpen
  \bibfield  {author} {\bibinfo {author} {\bibfnamefont {F.}~\bibnamefont
  {Buscemi}}\ and\ \bibinfo {author} {\bibfnamefont {M.}~\bibnamefont
  {Horodecki}},\ }\bibfield  {title} {\bibinfo {title} {Towards a unified
  approach to information-disturbance tradeoffs in quantum measurements},\
  }\href {https://doi.org/10.1142/S1230161209000037} {\bibfield  {journal}
  {\bibinfo  {journal} {Open Syst. Inf. Dyn.}\ }\textbf {\bibinfo {volume}
  {16}},\ \bibinfo {pages} {29} (\bibinfo {year} {2009})}\BibitemShut {NoStop}%
\bibitem [{Note3()}]{Note3}%
  \BibitemOpen
  \bibinfo {note} {{When $\Phi $ is a unitary operation, we have $L_{\protect
  \text {otm}}=0$ and $S(\Phi (|p_i\rangle \protect \!\langle p_i|))=0$.
  Because the entropic disturbance is invariant under the unitary operation
  $\Delta \chi =0$, we can say that the upper bound is tight for the unitary
  operation.}}\BibitemShut {Stop}%
\bibitem [{\citenamefont {Cerezo}\ \emph
  {et~al.}(2021{\natexlab{b}})\citenamefont {Cerezo}, \citenamefont {Sone},
  \citenamefont {Volkoff}, \citenamefont {Cincio},\ and\ \citenamefont
  {Coles}}]{cerezo2021barren}%
  \BibitemOpen
  \bibfield  {author} {\bibinfo {author} {\bibfnamefont {M.}~\bibnamefont
  {Cerezo}}, \bibinfo {author} {\bibfnamefont {A.}~\bibnamefont {Sone}},
  \bibinfo {author} {\bibfnamefont {T.}~\bibnamefont {Volkoff}}, \bibinfo
  {author} {\bibfnamefont {L.}~\bibnamefont {Cincio}},\ and\ \bibinfo {author}
  {\bibfnamefont {P.~J.}\ \bibnamefont {Coles}},\ }\bibfield  {title} {\bibinfo
  {title} {Cost-function-dependent barren plateaus in shallow quantum neural
  networks},\ }\href {https://doi.org/10.1038/s41467-021-21728-w} {\bibfield
  {journal} {\bibinfo  {journal} {Nat. Commun.}\ }\textbf {\bibinfo {volume}
  {12}},\ \bibinfo {pages} {1791} (\bibinfo {year}
  {2021}{\natexlab{b}})}\BibitemShut {NoStop}%
\bibitem [{\citenamefont {Schlosshauer}(2007)}]{schlosshauer2007decoherence}%
  \BibitemOpen
  \bibfield  {author} {\bibinfo {author} {\bibfnamefont {M.~A.}\ \bibnamefont
  {Schlosshauer}},\ }\href@noop {} {\emph {\bibinfo {title} {Decoherence: and
  the quantum-to-classical transition}}}\ (\bibinfo  {publisher} {Springer
  Science \& Business Media},\ \bibinfo {year} {2007})\BibitemShut {NoStop}%
\bibitem [{\citenamefont {Araki}\ and\ \citenamefont {Lieb}(1970)}]{Araki70}%
  \BibitemOpen
  \bibfield  {author} {\bibinfo {author} {\bibfnamefont {H.}~\bibnamefont
  {Araki}}\ and\ \bibinfo {author} {\bibfnamefont {E.~H.}\ \bibnamefont
  {Lieb}},\ }\bibfield  {title} {\bibinfo {title} {Entropy inequalities},\
  }\href {https://doi.org/10.1007/BF01646092} {\bibfield  {journal} {\bibinfo
  {journal} {Commun. Math. Phys.}\ }\textbf {\bibinfo {volume} {18}},\ \bibinfo
  {pages} {160} (\bibinfo {year} {1970})}\BibitemShut {NoStop}%
\bibitem [{\citenamefont {Wilde}\ \emph {et~al.}(2014)\citenamefont {Wilde},
  \citenamefont {Winter},\ and\ \citenamefont {Yang}}]{wilde2014strong}%
  \BibitemOpen
  \bibfield  {author} {\bibinfo {author} {\bibfnamefont {M.~M.}\ \bibnamefont
  {Wilde}}, \bibinfo {author} {\bibfnamefont {A.}~\bibnamefont {Winter}},\ and\
  \bibinfo {author} {\bibfnamefont {D.}~\bibnamefont {Yang}},\ }\bibfield
  {title} {\bibinfo {title} {Strong converse for the classical capacity of
  entanglement-breaking and hadamard channels via a sandwiched r{\'e}nyi
  relative entropy},\ }\href {https://doi.org/10.1007/s00220-014-2122-x}
  {\bibfield  {journal} {\bibinfo  {journal} {Commun. Math. Phys.}\ }\textbf
  {\bibinfo {volume} {331}},\ \bibinfo {pages} {593} (\bibinfo {year}
  {2014})}\BibitemShut {NoStop}%
\bibitem [{\citenamefont {Beigi}(2013)}]{beigi2013sandwiched}%
  \BibitemOpen
  \bibfield  {author} {\bibinfo {author} {\bibfnamefont {S.}~\bibnamefont
  {Beigi}},\ }\bibfield  {title} {\bibinfo {title} {Sandwiched r{\'e}nyi
  divergence satisfies data processing inequality},\ }\href
  {https://doi.org/10.1063/1.4838855} {\bibfield  {journal} {\bibinfo
  {journal} {J. Math. Phys.}\ }\textbf {\bibinfo {volume} {54}},\ \bibinfo
  {pages} {122202} (\bibinfo {year} {2013})}\BibitemShut {NoStop}%
\bibitem [{\citenamefont {Sagawa}(2020)}]{Sagawa20}%
  \BibitemOpen
  \bibfield  {author} {\bibinfo {author} {\bibfnamefont {T.}~\bibnamefont
  {Sagawa}},\ }\bibfield  {title} {\bibinfo {title} {Entropy, divergence, and
  majorization in classical and quantum thermodynamics},\ }\href
  {https://arxiv.org/abs/2007.09974} {\bibfield  {journal} {\bibinfo  {journal}
  {arXiv: 2007.09974}\ } (\bibinfo {year} {2020})}\BibitemShut {NoStop}%
\bibitem [{\citenamefont {Tomamichel}\ \emph {et~al.}(2010)\citenamefont
  {Tomamichel}, \citenamefont {Colbeck},\ and\ \citenamefont
  {Renner}}]{Tomamichel2010duality}%
  \BibitemOpen
  \bibfield  {author} {\bibinfo {author} {\bibfnamefont {M.}~\bibnamefont
  {Tomamichel}}, \bibinfo {author} {\bibfnamefont {R.}~\bibnamefont
  {Colbeck}},\ and\ \bibinfo {author} {\bibfnamefont {R.}~\bibnamefont
  {Renner}},\ }\bibfield  {title} {\bibinfo {title} {Duality between smooth
  min-and max-entropies},\ }\href
  {https://doi.org/https://doi.org/10.1109/TIT.2010.2054130} {\bibfield
  {journal} {\bibinfo  {journal} {IEEE Trans. Inf. Theory}\ }\textbf {\bibinfo
  {volume} {56}},\ \bibinfo {pages} {4674} (\bibinfo {year}
  {2010})}\BibitemShut {NoStop}%
\bibitem [{\citenamefont {Tomamichel}(2015)}]{Tomamichel2015quantum}%
  \BibitemOpen
  \bibfield  {author} {\bibinfo {author} {\bibfnamefont {M.}~\bibnamefont
  {Tomamichel}},\ }\href
  {https://doi.org/https://doi.org/10.1007/978-3-319-21891-5} {\emph {\bibinfo
  {title} {Quantum Information Processing with Finite Resources: Mathematical
  Foundations}}},\ Vol.~\bibinfo {volume} {5}\ (\bibinfo  {publisher}
  {Springer},\ \bibinfo {year} {2015})\BibitemShut {NoStop}%
\bibitem [{\citenamefont {Cappellini}\ \emph {et~al.}(2007)\citenamefont
  {Cappellini}, \citenamefont {Sommers},\ and\ \citenamefont
  {{\.Z}yczkowski}}]{Cappellini2007subnormalized}%
  \BibitemOpen
  \bibfield  {author} {\bibinfo {author} {\bibfnamefont {V.}~\bibnamefont
  {Cappellini}}, \bibinfo {author} {\bibfnamefont {H.-J.}\ \bibnamefont
  {Sommers}},\ and\ \bibinfo {author} {\bibfnamefont {K.}~\bibnamefont
  {{\.Z}yczkowski}},\ }\bibfield  {title} {\bibinfo {title} {Subnormalized
  states and trace-nonincreasing maps},\ }\href
  {https://aip.scitation.org/doi/10.1063/1.2738359} {\bibfield  {journal}
  {\bibinfo  {journal} {J. Math. Phys.}\ }\textbf {\bibinfo {volume} {48}},\
  \bibinfo {pages} {052110} (\bibinfo {year} {2007})}\BibitemShut {NoStop}%
\bibitem [{\citenamefont {Cerezo}\ \emph {et~al.}(2020)\citenamefont {Cerezo},
  \citenamefont {Poremba}, \citenamefont {Cincio},\ and\ \citenamefont
  {Coles}}]{Cerezo19Fidelity}%
  \BibitemOpen
  \bibfield  {author} {\bibinfo {author} {\bibfnamefont {M.}~\bibnamefont
  {Cerezo}}, \bibinfo {author} {\bibfnamefont {A.}~\bibnamefont {Poremba}},
  \bibinfo {author} {\bibfnamefont {L.}~\bibnamefont {Cincio}},\ and\ \bibinfo
  {author} {\bibfnamefont {P.~J.}\ \bibnamefont {Coles}},\ }\bibfield  {title}
  {\bibinfo {title} {Variational quantum fidelity estimation},\ }\href
  {https://quantum-journal.org/papers/q-2020-03-26-248/} {\bibfield  {journal}
  {\bibinfo  {journal} {Quantum}\ }\textbf {\bibinfo {volume} {4}},\ \bibinfo
  {pages} {248} (\bibinfo {year} {2020})}\BibitemShut {NoStop}%
\bibitem [{\citenamefont {Cappellaro}(2012)}]{cappellaro201222}%
  \BibitemOpen
  \bibfield  {author} {\bibinfo {author} {\bibfnamefont {P.}~\bibnamefont
  {Cappellaro}},\ }\href
  {https://ocw.mit.edu/courses/22-51-quantum-theory-of-radiation-interactions-fall-2012/pages/lecture-notes/}
  {\bibinfo {title} {22.51 {Q}uantum {T}heory of {R}adiation {I}nteractions
  {F}all 2012”, {M}assachusetts {I}nstitute of {T}echnology: {MIT}
  {O}pen{C}ourse{W}are}} (\bibinfo {year} {2012})\BibitemShut {NoStop}%
\bibitem [{\citenamefont {Sone}\ \emph {et~al.}(2020)\citenamefont {Sone},
  \citenamefont {Liu},\ and\ \citenamefont {Cappellaro}}]{Sone20a}%
  \BibitemOpen
  \bibfield  {author} {\bibinfo {author} {\bibfnamefont {A.}~\bibnamefont
  {Sone}}, \bibinfo {author} {\bibfnamefont {Y.-X.}\ \bibnamefont {Liu}},\ and\
  \bibinfo {author} {\bibfnamefont {P.}~\bibnamefont {Cappellaro}},\ }\bibfield
   {title} {\bibinfo {title} {Quantum {J}arzynski equality in open quantum
  systems from the one-time measurement scheme},\ }\href
  {https://doi.org/10.1103/PhysRevLett.125.060602} {\bibfield  {journal}
  {\bibinfo  {journal} {Phys. Rev. Lett.}\ }\textbf {\bibinfo {volume} {125}},\
  \bibinfo {pages} {060602} (\bibinfo {year} {2020})}\BibitemShut {NoStop}%
\bibitem [{\citenamefont {Sone}\ \emph {et~al.}(2022)\citenamefont {Sone},
  \citenamefont {Soares-Pinto},\ and\ \citenamefont {Deffner}}]{sone2022heat}%
  \BibitemOpen
  \bibfield  {author} {\bibinfo {author} {\bibfnamefont {A.}~\bibnamefont
  {Sone}}, \bibinfo {author} {\bibfnamefont {D.~O.}\ \bibnamefont
  {Soares-Pinto}},\ and\ \bibinfo {author} {\bibfnamefont {S.}~\bibnamefont
  {Deffner}},\ }\bibfield  {title} {\bibinfo {title} {Exchange fluctuation
  theorems for strongly interacting quantum pumps},\ }\href
  {https://arxiv.org/abs/2209.12927} {\bibfield  {journal} {\bibinfo  {journal}
  {arXiv preprint arXiv:2209.12927}\ } (\bibinfo {year} {2022})}\BibitemShut
  {NoStop}%
\bibitem [{\citenamefont {Beyer}\ \emph {et~al.}(2020)\citenamefont {Beyer},
  \citenamefont {Luoma},\ and\ \citenamefont {Strunz}}]{Beyer2020}%
  \BibitemOpen
  \bibfield  {author} {\bibinfo {author} {\bibfnamefont {K.}~\bibnamefont
  {Beyer}}, \bibinfo {author} {\bibfnamefont {K.}~\bibnamefont {Luoma}},\ and\
  \bibinfo {author} {\bibfnamefont {W.~T.}\ \bibnamefont {Strunz}},\ }\bibfield
   {title} {\bibinfo {title} {Work as an external quantum observable and an
  operational quantum work fluctuation theorem},\ }\href
  {https://link.aps.org/doi/10.1103/PhysRevResearch.2.033508} {\bibfield
  {journal} {\bibinfo  {journal} {Phys. Rev. Research}\ }\textbf {\bibinfo
  {volume} {2}},\ \bibinfo {pages} {033508} (\bibinfo {year}
  {2020})}\BibitemShut {NoStop}%
\end{thebibliography}%

\end{document}